\documentclass[12pt,twoside]{amsart}
\usepackage{amsmath}
\usepackage{amssymb}
\usepackage{mathrsfs}         
\usepackage{yhmath}			  
\usepackage{subcaption}
\usepackage{pgfplots}
\pgfplotsset{compat=newest}
\usepgfplotslibrary{external}

\usepackage{tikz,tkz-euclide}
\usetikzlibrary{arrows,decorations.pathmorphing,decorations.markings,backgrounds,positioning,fit,petri,patterns,calc,intersections}
\usetkzobj{all} 
\usetikzlibrary{3d}
\newcommand\pointlat[1]{ ({1.5*cos(#1)},{1.5*sin(#1)}) }
\newcommand\pointlatl[1]{ ({-1.5*cos(#1)},{1.5*sin(#1)}) }
\newcommand\pointlon[1]{ ({1.5*sin(#1)},{1.5*cos(#1)}) }
\newcommand\pointlond[1]{ ({1.5*sin(#1)},{-1.5*cos(#1)}) }
\newcommand\latitude[1]{%
  \draw (#1:1.5) arc (0:-180:{1.5*cos(#1)} and {0.2*cos(#1)});
  \draw[dashed] (#1:1.5) arc (0:180:{1.5*cos(#1)} and {0.2*cos(#1)});
}
\newcommand\longitude[1]{%
  \draw[] ({1.5*sin(#1)},{1.5*cos(#1)}) arc (90:270:{0.2*cos(#1)} and {1.5*cos(#1)});
  \draw[dashed] ({1.5*sin(#1)},{1.5*cos(#1)})  arc (90:270:{-0.2*cos(#1)} and {1.5*cos(#1)});
}

\newcommand\axes{
\draw[->] (1.5,0) -- (2,0) node[anchor=north west] {$ \gamma_1$};
\draw[->] (-0.2,-0.2) -- (-0.5,-0.5) node[anchor=north  east] {$ \gamma_3$};
\draw[->] (0,1.5) -- (0,2) node[anchor=north west] {$ \gamma_2$};
}
\newcommand\circlez{(0,0) circle (1.5cm)}
\newcommand\regionh[1]{ \pointlat{#1} arc (0:-180:{1.5*cos(#1)} and {0.2*cos(#1)})-- \pointlatl{#1} arc ({#1}:-{#1}:-1.5cm and 1.5cm)--\pointlatl{-#1} arc (0:-180:{-1.5*cos(#1)} and {0.2*cos(#1)})-- \pointlat{-#1} arc (-{#1}:{#1}:1.5cm and 1.5cm)}
\newcommand\regionv[1]{  \pointlon{-#1} arc (90:270:{0.2*cos(#1)} and {1.5*cos(#1)})-- \pointlond{-#1} arc ({-90-#1}:{-90+#1}:1.5 and 1.5)--\pointlond{#1} arc (-90:-270:{0.2*cos(#1)} and {1.5*cos(#1)})--\pointlon{#1} arc ({90-#1}:{90+#1}:1.5 and 1.5)}

\usepackage{xcolor}

\usepackage[colorlinks=true, urlcolor=blue,bookmarks=true,bookmarksopen=true, citecolor=blue]{hyperref}

\numberwithin{equation}{section}

\newtheorem{theorem}{Theorem}[section]

\newtheorem{corollary}[theorem]{Corollary}

\newtheorem{lemma}[theorem]{Lemma}
\newtheorem{prop}[theorem]{Proposition}
\newtheorem{proposition}[theorem]{Proposition}

\def \begineq{\begin{equation}}
\def \endeq{\end{equation}}

\def \bb{\mathbb}

\def \mc{\mathcal}

\def \QQ{{\bb{Q}}}
\def \RR{{\bb{R}}}

\def \ZZ{{\bb{Z}}}

\def \({\left(}
\def \){\right)}
\def \<{\left\langle}
\def \>{\right\rangle}

\def \bar{\overline}

\def \inter{\cap}

\def \union{\cup}
\def \vargeq{\geqslant}
\def \varleq{\leqslant}

\def \xto{\xrightarrow}

\def \const{{const}}

\usepackage{graphicx}

\begin{document}
 \title{Suslov problem with the Klebsh-Tisserand potential}
 \author{Shengda Hu, Manuele Santoprete}
 \abstract
 In this paper, we study a nonholonomic mechanical system, namely the Suslov problem with the Klebsh-Tisserand potential.  We analyze the topology of the level sets defined by the integrals in two ways: using an explicit construction and as a consequence of the Poincar\'e-Hopf theorem. We   describe the flow on such manifolds.
 \endabstract
 \maketitle
 
 \section{Introduction}\label{sect:preintro}
 A Hamiltonian  system on a $2n$-dimensional symplectic manifold is an called completely  integrable if it admits $n$ independent integrals of motion in involution. For such systems, if the common level sets  $ S _k $  of the integrals are compact, by the Liouville-Arnold theorem, the $ S _k $  are  invariant tori of dimension $n$, and the flow on the tori is isomorphic to a linear flow. The system is super-integrable if there are more than $n$ independent integrals of motions, and the invariant tori are of dimension less than $n$.
 
 In the present paper, we are concerned with a family of dynamical systems, the so called Suslov's problem, that are not Hamiltonian, but exhibit important features of integrable and super-integrable Hamiltonian system. As first formulated in \cite{suslov1946theoretical}, it describes the  dynamics of a rigid body with a fixed point immersed in a potential field and subject to a nonholonomic constraint that forces the angular velocity component along a given direction in the body to vanish.
 Our analysis shows that such systems have invariant tori carrying linear flows, as well as other types of invariant submanifolds carrying generically periodic flows.
 
 The topology of invariant submanifolds of this problem have been studied by Tatarinov \cite{tatarinov1985,tatarinov1988} using surgery methods, and Fernandez-Bloch-Zenkov \cite{fernandez2014geometry} using a generalization of the Poincar\'e-Hopf theorem to manifold with boundary together with some detailed information about the geometry of the problem. It was shown that the invariant submanifolds of this problems can be surfaces of genus between zero and five.

 We will provide two further approaches for understanding the topology of the submanifolds, as well as the flows. The first is a direct construction that uses a Morse theoretic reasoning and in our opinion provides a better understanding of the geometry of the problem than the other approaches. The second is an application of the classical Poincar\'e-Hopf theorem for manifolds without boundary and requires only knowledge of the number of connected components of the manifold. Furthermore, we give a detailed analysis of the flow on the invariant submanifolds and find that, for certain values of the parameters, the system admits an additional integral of motion. The information thus obtained leads to concrete understanding of the physical motion of the problem. 

Suslov's system is an  example of an important  class of nonholonomic systems, namely, the quasi-Chaplygin systems introduced in  \cite{fedorov2006quasi}.  This  system is Hamiltonizable in a very precise sense, that we describe below. 
%
  A non-holonomic system $ (Q, \mathcal{D} , L) $   consists of a configuration manifold $ Q $, a Lagrangian $ L: TQ\to \mathbb{R}  $ and a non-integrable smooth distribution $ \mathcal{D} \subset T Q $ describing kinematic constraints. The equations of motion   are determined by the Lagrange-d'Alembert principle supplemented by the condition that the velocities are in $ \mathcal{D} $, explicitely we have 
\begin{equation}\label{eqn:Lagrange-dAlambert}
    \sum _{ i = 1 } ^{ n } \left( \frac{ \partial L } { \partial x _i } - \frac{ d } { dt } \frac{ \partial L } { \partial \dot x _i } \right) \eta  _i = 0, \qquad \text{for all} \quad \eta  \in \mathcal{D} _x.
\end{equation} 
A {\bf Chaplygin system} is a non-holonomic system $ (Q, \mathcal{D}, L)$  where $ Q $ has a principal bundle structure $ Q \to Q/ G $ with respect to a Lie group $ G  $, with a principal connection, and the distribution $ \mathcal{D} $ is the horizontal bundle of the connection. Therefore, given a vector $ Y  \in T _x Q $, we have the decomposition $ Y  = Y _h + Y _v $, with $ Y _h \in \mathcal{D} _x $, and $ X _v \in \mathcal{V}  _x $, where $ \mathcal{V}  $   is the vertical bundle, and $ \mathcal{V} _x $  is the tangent space to the fiber at $ x $. In addition one must require that the Lagrangian $ L $ is $ G $-invariant. 

A non-holonomic system is {\bf quasi-Chaplygin} if,  $ Q $ has a principal bundle structure $ Q \to Q/ G $ with respect to a Lie group $ G  $, but the connection is singular  on some $ G $-invariant subvariety $ S \subset Q $, that is,  for $ x \in S $  the sum $ \mathcal{D}  _x + \mathcal{V}  _x $ does not span the tangent space $ T _x Q $. Quasi-Chaplygin systems can be regarded as  usual Chaplygin systems in  $ Q \setminus S $, inside the set $ S$, however, they behave differently (see \cite{fedorov2006quasi}). 
It was shown in \cite{fedorov2006quasi} that Suslov's problem is a quasi-Chaplygin system with $ Q = SO (3) $ and $ G = SO (2) $. 

For Chaplygin systems the constrained Lagrangian $ L _c (x, \dot x)  = L (x, \dot x  _h) $ induces a  Lagrangian 
$  l : TM \to \mathbb{R}  $ via the identification $ TM \approx \mathcal{D} /G $ (note that $ \mathcal{D} /G $  
has the structure of a vector bundle with base space $ Q/ G $ and fiber $ \mathbb{R}    ^k $, with $ k = \dim (Q/G) $).
Suppose the Legendre transform exists for $ l $, which in local coordinates $ q $ on $ M $ is given by  $ (q , \dot q) \to \left( q , \frac{ \partial l } { \partial \dot q }\right) $. Under the Legendre transform, the system of equations \eqref{eqn:Lagrange-dAlambert} gives rise to a first order dynamical system on $ T ^\ast M $ with corresponding vector field $ X _{ nh } = \{ \cdot , H \} _{ AP }  $ for some function  $ H : T ^\ast M  \to \mathbb{R}  $, where $ \{ \cdot , \cdot \} _{ AP } $ is an almost Poisson bracket (i.e. a skew-symmetric bilinear operation on functions that satisfies the Leibniz identity but  fails to satisfy the Jacobi identity). 
We say that the system is {\bf Chaplygin Hamiltonizable}, if there is a nonvanishing function $ f :  Q/G\to \mathbb{R}   $ such that 
\[
    X _{ nh } = f (q) X _H 
\]
with  $ X _H = \{ \cdot , H \} $,  and $ \{ \cdot , \cdot \}  = \frac{ 1 } { f } \{ \cdot , \cdot \} _{ AP } $ is a Poisson bracket.

For a quasi-Chaplygin system, if there is a function $ f :  Q/G\to \mathbb{R}$, nonvanishing in $ Q \setminus S $, and   such that $ X _{ nh } = f (q) X _H $
for some $ H : T ^\ast M  \to \mathbb{R}  $,  we call the system {\bf quasi-Chaplygin Hamiltonizable}. The Suslov problem considered in this article is quasi-Chaplygin Hamiltonizable with $f = \gamma_3$ (\cite{fedorov2006quasi, fernandez2014geometry}).
For this type of systems, since the multiplier $f $ has zeroes, one hypothesis of  Theorem 1 in  \cite{kozlov2002integration} fails and thus the topology of invariant manifolds may differ from tori.
Here we give an explicit description of the invariant manifolds.  We expect that the more explicit approach taken in this article may be able to shed light on more general quasi-Chaplygin Hamiltonizable systems.

\vskip 0.5cm
The paper is organized as follows. In Section 2 we  present an overview of the Suslov's problem in the Klebsh-Tisserand case. We give an elementary derivation of the equations of motion (see \cite{fedorov2006quasi} for a derivation based on a Lagrangian approach). In Section 3 we give an explicit construction that allows us to determine the topology of the level surfaces $ S _k $. Then, in Section 4 we study the flow of the system on the surfaces $ S _k $, and we find that, for some parameter values there is one additional integral of motion. 
We then find  and classify the critical pints of the vector field. We use the Poincar\'e-Hopf theorem to give an alternative way to determine the  topology of the surfaces. In Section 5 we use the topology of $ S _k $ and the results on the dynamics of the problem to describe how the rigid body moves in the three dimensional physical space. 
 \section{Suslov's Problem with a Klebsh-Tisserand Potential} \label{sect:intro}
The Suslov problem describes the motion of a rigid body with a fixed point  subject to a nonholonomic constraint. 
Wagner \cite{wagner1941} suggested the following implementation of Suslov's model. He considered a rigid body, with a fixed point $ O $, moving inside a spherical shell. The rigid body is attached at $ O $ with  a spherical hinge so that it can turn around this point. 
The  nonholonomic constraint is realized  by considering two rigid caster wheels  attached to the rigid body  by a rod (see figure \ref{fig:Suslov}). 
These wheels force the angular velocity component along a  direction orthogonal to the rod  to vanish.

\begin{center} 
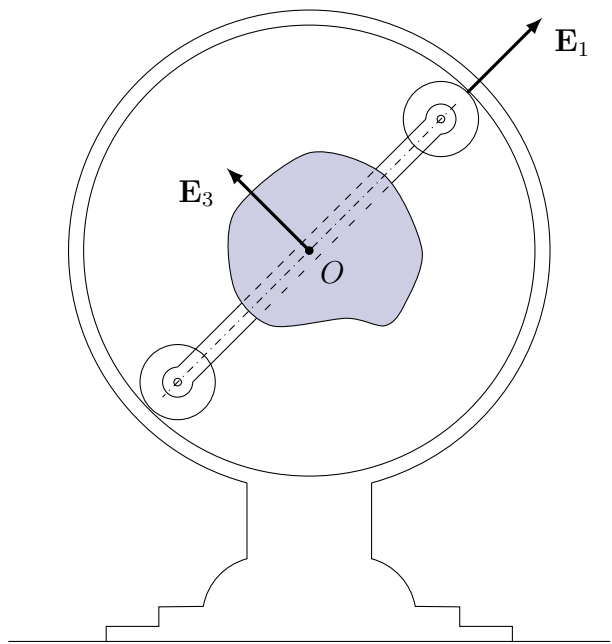
\begin{figure}[h]
\begin{tikzpicture}[>=latex]
\draw (0,0) circle (3cm);

\draw[] ([shift=(255:3.2cm)]0,0) arc (255:-75:3.2cm);
\draw ({3.2*cos(-75)},{3.2*sin(-75)})--({3.2*cos(-75)},-4.1)--({3.2*cos(-75)},-4.1) arc (75:10:0.8cm)--(2,-4.74)--(2,-5)--(2.7,-5)--(2.7,-5.2);
\draw ({3.2*cos(255)},{3.2*sin(255)})--({3.2*cos(255)},-4.1)--({3.2*cos(255)},-4.1) arc (105:170:0.8cm)--(-2,-4.74)--(-2,-5)--(-2.7,-5)--(-2.7,-5.2);
\draw (-4,-5.2)--(4,-5.2);
\draw (1.75,1.75) circle (0.5cm);
\draw (-1.75,-1.75) circle (0.5cm);
\draw (1.75,1.75) circle (0.05cm);
\draw (-1.75,-1.75) circle (0.05cm);

\draw[black] (1.75,1.55) -- (-1.55,-1.75);
\draw[black] (1.55,1.75) -- (-1.75,-1.55);


\draw[smooth cycle,fill=blue!30!gray!30!] plot[tension=0.5] coordinates{(-1,-1/2) (-1,1/2) (0,1.3)(2/2,2/2) (3/2,0)(4.5/2-1,0.5/2-1)(2/2,-1)(1/2,0.2/2-1)(-1/2,-1)};
\draw[] ([shift=(180:0.2cm)]1.75,1.75) arc (180:-90:0.2cm);
\draw[] ([shift=(90:0.2cm)]-1.75,-1.75) arc (90:360:0.2cm);
\draw[black,loosely dashed] (1.05,0.85) -- (-0.75,-0.95);
\draw[black,dashed] (0.85,1.05) -- (-0.95,-0.75);
\draw[black, dashdotted] (1.95,1.95) -- (-1.95,-1.95);
\draw[fill=black] (0,0) circle (0.05cm)  node[anchor=north west] {$O$};




\draw[->,very thick] (2.1,2.1)--(3.1,3.1)  node[anchor=north west] {$ \mathbf{E}  _1 $ };
\draw[->,very thick] (0,0)--(-1.1,1.1) node[anchor=north east] {$\mathbf{E}  _3 $};
\end{tikzpicture}
\caption{An implementation  of the Suslov problem suggested by Wagner.\label{fig:Suslov}}
\end{figure}
\end{center}

In this section we give an elementary derivation of the equations of motion for the Suslov's problem. These equations can also be obtained from a Lagrangian \cite{fedorov2006quasi}.
We begin by discussing the Euler equations for a rigid  body without constraint. Denote by $ \{ \mathbf{e}_1 , \mathbf{e}_2, \mathbf{e}_3 \} $ a right handed orthonormal basis of $ \mathbb{R}  ^3 $, called the \emph{spatial} frame. The coordinates of a point $P$ is the spacial frame is denoted $ \mathbf{x}$, which is also called the \emph{spatial} vector of $P$. Let $ \{ \mathbf{E}_1 , \mathbf{E}_2 , \mathbf{E}_3 \} $ be a right handed orthonormal frame, the \emph{body} frame, defined by the three principal axis and let $ \mathbf{X} $ denote the \emph{body} vector of the point $P$ in this frame. We have
\[\mathbf{x} = \mathbf{R} \mathbf{X}, \text{ where } \mathbf{R} \text{ is a rotation matrix}\]
Let $ \boldsymbol{\omega } $, $ \boldsymbol{\pi } $  and $ \boldsymbol{\tau} $ be the \emph{spacial} vector of angular velocity,  angular momentum, and  torque, respectively. The corresponding \emph{body} vectors are then given by
\[\boldsymbol{\Omega} = \mathbf{R}^{-1} \boldsymbol{\omega}, \boldsymbol{\Pi } = \mathbf{R} ^{ - 1 } \boldsymbol{\pi } \text{ and } \boldsymbol{T } = \mathbf{R} ^{ - 1 } \boldsymbol{\tau}\] 
Note that if   $ \mathbb{I}  = \operatorname{diag} (I_1 , I_2 , I_3) $ denotes the body inertia tensor then we can also write $ \boldsymbol{\Pi } = \mathbb{I}   \boldsymbol{\Omega } $.  
Now $\boldsymbol{\tau} = \boldsymbol{\dot\pi}$, the second cardinal equation of dynamics, can be rewritten as
\[
    \mathbf{R} \boldsymbol{T} = \boldsymbol{\tau} = \frac{ d } { d t }\left(  { \mathbf{R}  \boldsymbol{\Pi } }\right) =  {\bf \dot R } \boldsymbol{ \Pi }+ {\bf R } \boldsymbol{\dot \Pi } \Longrightarrow \boldsymbol{\dot \Pi } =   -{\bf  R } ^{ - 1 } {\bf \dot R } \boldsymbol{ \Pi } + \mathbf{T} 
\]
Let $\boldsymbol{ \hat \Omega } = \mathbf{R} ^{ - 1 } {\bf \dot R }$ be the \emph{hat map}, then we have
\[ \boldsymbol{ \hat \Omega } \boldsymbol{ \Pi } = \boldsymbol{\Omega } \times \boldsymbol{\Pi }\]
which gives
\[
    \boldsymbol{\dot \Pi }  = \boldsymbol{\Pi } \times \boldsymbol{\Omega }  + \mathbf{T}  
\]

Let $ \boldsymbol{\alpha }$, $ \boldsymbol{ \beta }$, and  $ \boldsymbol{\gamma } $ be the body vectors of $ \mathbf{e}_1, \mathbf{e}_2, $  and $ \mathbf{e}_3 $ respectively, e.g.
\[\boldsymbol{\gamma } = {\bf R } ^{ - 1 }  {\bf e }_3 \]
Then we obviously have
\[\| \boldsymbol{\gamma } \| ^2 = \gamma_1 ^2 + \gamma_2 ^2 + \gamma_3 ^2 = 1 .\]
Suppose the rigid body is placed in a force field with potential energy 
\[ u ( \mathbf{X})=  u \left(  \left\langle \mathbf{X} , \boldsymbol{ \gamma  } \right\rangle\right),\] 
where $ \left\langle\, , \,\right\rangle $ is the Euclidean inner product. 
Then the total potential energy of the rigid body is 
\[
 U (\boldsymbol{\gamma }) =  \int_{ B }  u \left( \left\langle \mathbf{X} , \boldsymbol{ \gamma  } \right\rangle\right)  \,d ^3 \mathbf{X}   
 \]
and the total  body force and torque  are 
\[\mathbf{F} = -\int_{ B } \frac{ \partial u } { \partial \mathbf{X} } \,d ^3 \mathbf{X},\quad \mbox{and }  \mathbf{T} =- \int_{ B }   
    \frac{ \partial u } { \partial \mathbf{X} }\times \mathbf{X}  \,d ^3 \mathbf{X},
\]
respectively. We have
\begin{align*} 
        \frac{ \partial u } { \partial \mathbf{X} } \times \mathbf{X} & =   \frac{ \partial u } { \partial \left\langle \mathbf{X} , \boldsymbol{\gamma  } \right\rangle } \frac{ \partial \left\langle \mathbf{X} , \boldsymbol{\gamma  } \right\rangle } { \partial \mathbf{X} } \times \mathbf{X}  = - \left( \frac{ \partial u } { \partial \left\langle \mathbf{X} , \boldsymbol{\gamma  } \right\rangle } \mathbf{X} \right) \times \boldsymbol{\gamma},
\end{align*} 
and similarly
\[
    \frac{ \partial u } { \partial \boldsymbol{\gamma } }= \frac{ \partial u } { \partial \left\langle \mathbf{X} , \boldsymbol{\gamma }\right\rangle  } \frac{   \partial \left\langle \mathbf{X} , \boldsymbol{\gamma }\right\rangle } { \partial \boldsymbol{\gamma } } =  \frac{ \partial u } { \partial \left\langle \mathbf{X} , \boldsymbol{\gamma }\right\rangle  }\mathbf{X}  
\]
which gives
\[
 \frac{ \partial u } { \partial \mathbf{X} } \times \mathbf{X} = - \frac{ \partial u } { \partial \boldsymbol{\gamma } } \times \boldsymbol{\gamma }.
\]
Since $ \boldsymbol{\gamma } $ does not depend on $ \mathbf{X} $, integrating yields the following expression for the torque:
\[
    \mathbf{T} = \frac{ \partial U } { \partial \boldsymbol{\gamma } } \times \boldsymbol{\gamma }
\]
and thus the dynamics equations are 
\[
    \boldsymbol{\dot \Pi } =  \boldsymbol{\Pi } \times \boldsymbol{\Omega }  +    \frac{ \partial U } { \partial \boldsymbol{\gamma} } \times \boldsymbol{\gamma }, \quad   \boldsymbol{\dot \gamma }=  \boldsymbol{\gamma } \times \boldsymbol{\Omega } 
\]

Let $ \mathbf{A} $ be a fixed unit body vector, and consider \emph{Suslov's nonholonomic constraint}
\[
    \left\langle \boldsymbol{\Omega }, \mathbf{A}\right\rangle  = 0
\]
Subjecting the rigid body to this constraint is equivalent to adding a torque $ \lambda \mathbf{A} $. The rigid body subject to this torque moves according to 
\[
    \boldsymbol{\dot \Pi } =  \boldsymbol{\Pi } \times \boldsymbol{\Omega }  +    \frac{ \partial U } { \partial \boldsymbol{\gamma} } \times \boldsymbol{\gamma } + \lambda \mathbf{A} , \qquad   \boldsymbol{\dot \gamma }=  \boldsymbol{\gamma } \times \boldsymbol{\Omega }, \qquad     \left\langle \boldsymbol{\Omega }, \mathbf{A}\right\rangle  = 0,
\]
where the first equation can also be written as 
\begin{equation} \label{eqn:multiplier}
    \mathbb{I}  \boldsymbol{\dot \Omega  } =  \mathbb{I} \boldsymbol{\Omega  } \times \boldsymbol{\Omega }  +    \frac{ \partial U } { \partial \boldsymbol{\gamma} } \times \boldsymbol{\gamma } + \lambda \mathbf{A}. 
\end{equation} 
Differentiating the constraint  gives $  \left\langle \mathbf{A} , \boldsymbol{\dot \Omega }\right\rangle  = 0$, and substituting \eqref{eqn:multiplier} into the constraint equation yields 
\[
    \left\langle \mathbf{A}, \mathbb{I} ^{ - 1 } \left(  \mathbb{I} \boldsymbol{\Omega  } \times \boldsymbol{\Omega }  +    \frac{ \partial U } { \partial \boldsymbol{\gamma} } \times \boldsymbol{\gamma } + \lambda \mathbf{A} \right) \right\rangle = 0 
\]
 and hence 
 \[
     \lambda = -\frac{    \left\langle \mathbf{A}, \mathbb{I} ^{ - 1 } \left(  \mathbb{I} \boldsymbol{\Omega  } \times \boldsymbol{\Omega }  +    \frac{ \partial U } { \partial \boldsymbol{\gamma} } \times \boldsymbol{\gamma }  \right) \right\rangle } { \left\langle \mathbf{A} , \mathbb{I} ^{ - 1 } \mathbf{A}\right\rangle   }.  
 \]

 Let us consider the case $ \mathbf{A} = {\bf E }_3 $, so that the constraint is simply $ \Omega_3 = 0 $. 
 Then the equation of the rigid body subject to Suslov's nonholonomic constraint are 
\begin{align*} 
    I_1 \dot \Omega_1 = \gamma_2 \frac{ \partial U } { \partial \gamma_3 } - \gamma_3 \frac{ \partial U } { \partial \gamma_2 },  \quad I_2 \dot \Omega_2 = \gamma_3 \frac{ \partial U } { \partial \gamma_1 } - \gamma_1 \frac{ \partial U } { \partial \gamma_3 },
\end{align*}
and 
\[
    \dot \gamma_1 = - \gamma_3 \Omega_2 , \quad \dot \gamma_2 = \gamma_3 \Omega_1 , \quad \dot \gamma_3 =  \gamma_1 \Omega_2 - \gamma_2 \Omega_1,
\]
where we omitted the equation for $ \dot \Omega  _3 $ since it  is used only to determine $ \lambda $ and can be omitted. 
We consider the Klebsh-Tisserand case of the Suslov problem, i.e.
\[ U (\gamma) = \frac{1}{2} \left( B_1 \gamma_1 ^2 + B_2 \gamma_2 ^2\right)\]
Then  the equations of motion in terms of the momenta are given by  
\begin{align*} 
     \dot \Pi _1 & = -B_2 \gamma_2 \gamma_3 \\
     \quad  \dot \Pi  _2 & = B_1 \gamma_1 \gamma_3\\
    \dot \gamma_1 & = - \gamma_3 \frac{ \Pi_2 } { I_2 }  \\
    \quad \dot \gamma_2 & = \gamma_3 \frac{ \Pi_1 } { I_1 }  \\
    \quad \dot \gamma_3 & =  \gamma_1 \frac{ \Pi_2 } { I_2 }  - \gamma_2 \frac{ \Pi_1 } { I_1 } . 
\end{align*} 
The following functions are easily seen to be integrals of motion of the equations above
\[
    F_1 : = \frac{ \Pi_1 ^2 } { I_1 } + B_2 \gamma_2 ^2, \quad F_2 := \frac{ \Pi _2 ^2 } { I_2 } + B_1 \gamma_1 ^2. 
\]

 Understanding of the Suslov problem now reduces to understanding of the flows on the level surfaces $F_1^{-1}(K_1) \inter F_2^{-1}(K_2)$ defined by the integrals of motion.
 It is convenient to perform the following change of variables:
\[
    m_1 = -\frac{ \Pi_2 } { I_2 }, \quad b_1 = \frac{ B_1 } { I_2 }, \quad k_1 = \frac{ K_1 } { I_2 };\quad  
    m_2 = -\frac{ \Pi_1 } { I_1 }, \quad b_2 = \frac{ B_2 } { I_1 }, \quad k_2 = \frac{ K_2 } { I_1 }.
\]
and consider the system as defined in $\RR^5$, with coordinates $(m_1, m_2, \gamma_1, \gamma_2, \gamma_3)$, subject to the restriction
\[\gamma_1^2 + \gamma_2^2 + \gamma_3^2 = 1\]
Then the integrals of motion become $ f_1 = m_1 ^2 + b_1 \gamma_1 ^2 $ and $ f_2 = m_2 ^2 + b_2 \gamma_2 ^2 $. We write down the equations defining the level surfaces $S_k := f_1^{-1}(k_1) \inter f_2^{-1}(k_2)$:
 \begin{equation}\label{eq:suslovsurface}
  \begin{cases}
   \displaystyle{m_1^2 + b_1\gamma_1^2} = k_1 \\
   \displaystyle{m_2^2 + b_2\gamma_2^2} = k_2 \\
   \gamma_1^2 + \gamma_2^2 + \gamma_3^2 = 1
  \end{cases}
 \end{equation}
 We also write the equations of motion in the new coordinates:
 \begin{equation}\label{eq:suslovflow}
  \begin{cases}
   \dot m_1 = -b_1 \gamma_1\gamma_3 \\
   \dot m_2 = b_2 \gamma_2 \gamma_3 \\
   \displaystyle{\dot \gamma_1 = m_1 \gamma_3} \\
   \displaystyle{\dot \gamma_2 = -m_2 \gamma_3} \\
   \displaystyle{\dot \gamma_3 = \gamma_2 m_2 - \gamma_1 m_1}.
  \end{cases}
 \end{equation}
 We denote by $ X = (X_1 , X_2 , X_3 , X_4 , X_5) : \mathbb{R}  ^5 \to \mathbb{R}  ^5 $, the vector field associated with the equations above.
 
\section{Topology of the level surfaces \texorpdfstring{$S_k$}{Sk} via an explicit construction}  \label{sect:Suslovproblemdim3}
In this section we want to describe the topology of the level surfaces $ S _k $ by giving an explicit construction. Our method differs from the surgery approach pioneered by Tatarinov \cite{tatarinov1985,tatarinov1988}. Note that while it is difficult to find the original work of Tatarinov, an exposition of his method can be found in \cite{fomenko2012visual} and \cite{fernandez2014geometry}.  

We start by finding the values of $ k = (k _1 , k _2) $ for which $ S _k $  is non-singular.
\begin{lemma}\label{lemma:singularlevels}

  The subspace $S_k$ is a smooth manifold of dimension $2$ iff $k_1k_2 \neq 0$ and all of the following holds: $k_1 \neq b_1$, $k_2 \neq b_2 $ and $\displaystyle{\frac{k_1}{b_1} + \frac{k_2}{b_2} \neq 1}$.
 \end{lemma}
 \begin{proof}
   The matrix formed by the gradients of the defining equations is
\begin{equation}\label{eqn:gradients}  2
    \begin{pmatrix}
     m_1 & 0 & b_1\gamma_1 & 0 & 0\\
     0 & m_2 & 0 & b_2\gamma_2 & 0\\
     0 & 0 & \gamma_1 & \gamma_2 & \gamma_3\\ 
    \end{pmatrix}
\end{equation} 
   and $S_k$ is smooth iff matrix above to have maximal rank $3$ at all points on $S_k$. Obviously, the matrix is of full rank when $m_1 m_2 \neq 0$; while it is degenerate when $k_1k_2 = 0$.
  
   For $m_1 = m_2 = 0$, we have, $b_1\gamma_1^2 = k_1$ and $b_2 \gamma_2^2 = k_2$. It follows that
   $$\gamma_3^2 = 1 - \frac{k_1}{b_1} - \frac{k_2}{b_2}$$
   The rank $3$ condition requires $\displaystyle{\frac{k_1}{b_1} + \frac{k_2}{b_2} \neq 1}$.
  
   For $m_1 = 0$ but $m_2 \neq 0$, we have $b_1\gamma_1^2 = k_1$. Rank $3$ condition requires one of $\gamma_2$ or $\gamma_3$ is non-zero. We have $\gamma_2 = \gamma_3 = 0 \iff b_1 = k_1$, thus full rank implies $k_1 \neq b_1$.
  
   The case $m_1 \neq 0$ but $m_2 = 0$ gives $k_2 \neq b_2$.
  \end{proof}
  A simple consequence of the Lemma is that the topology of the level sets  $ S _k $ can only change at values of $ k $ where $ S _k $ is singular. This is made precise in the following corollary.  
 \begin{corollary}\label{coro:5regions} 
  The first quadrant in the $(k_1, k_2)$-plane is divided into $5$ regions  by
  $$k_1 = b_1, k_2 = b_2 \text{ and } \frac{k_1}{b_1} + \frac{k_2}{b_2} = 1$$
  The subspace $S_k$ has the same topological type for $k$ in each region (c.f. Figure \ref{fig:bifurcations}).
 \end{corollary}
    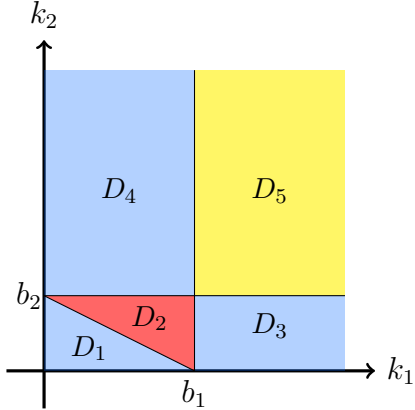
\begin{figure}[!ht]
   \begin{minipage}[c]{0.33\textwidth}
    \begin{tikzpicture}[scale=2]
     \draw[very thick,->] (-0.25,0) -- (2.2,0) node[right] {$k_1$};
     \draw[very thick,->] (0,-0.25) -- (0,2.2) node[above] {$k_2$};

     \fill[yellow,opacity=0.6] (1,0.5) -- (2,0.5) -- (2,2) -- (1,2)--cycle;
     \fill[red,opacity=0.6] (1,0) -- (1,0.5) -- (0,0.5) --cycle;
     \fill[blue!60!cyan,opacity=0.3] (1,0) -- (2,0) -- (2,0.5) -- (1,0.5)--cycle;
     \fill[blue!60!cyan,opacity=0.3] (0,0.5) -- (0,2) -- (1,2) -- (1,0.5)--cycle;
     \fill[blue!60!cyan,opacity=0.3] (0,0) -- (0,0.5) -- (1,0) --cycle;

     \draw (1,0) -- (0,0.5);
     \draw (1,0) --  (1,2);
     \draw (0,0.5) --  (2,0.5);
     \node at (1.5,1.2) {\small$D_5$ };
     \node at (0.7,0.35) {\small$D_2$ };
     \node at (0.3,0.15) {\small$D_1$ };
     \node at (1.5,0.3) {\small$D_3$ };
     \node at (0.5,1.2) {\small$D_4$ };
     \node at (1,-0.15) {$ b_1   $ };
     \node at (-0.1,0.5) {$ b_2   $ };
    \end{tikzpicture}
   \end{minipage}
   \begin{minipage}[c]{0.65\textwidth}
    \caption{\label{fig:bifurcations} 
     \small $ S_k $ is a smooth two-manifold (possibly disconnected) when $(k_1, k_2)$ lies in the interiors of the shaded regions. For example, we will see that in the region $D_5$ shaded yellow, $S_k = 4S^2 = S^2 \sqcup S^2 \sqcup S^2 \sqcup S^2$.
    }
   \end{minipage}
   \end{figure}
 
 The level surfaces $S_k$ are complete intersections. In \eqref{eq:suslovsurface}, the first two equations define a $2$-torus $T^2_k \subset \RR^4$, with coordinates $(m_1, m_2, \gamma_1, \gamma_2)$, while the last equation defines the unit $2$-sphere $S^2 \subset \RR^3$, with coordinates $(\gamma_1, \gamma_2, \gamma_3)$. It is therefore natural to study the level sets $ S _k $ by analyzing their projections onto these well understood surfaces. The projection onto the torus will be studied in the next subsection and it will be crucial in determining the topology of the level surfaces. The projection onto the  unit $2$-sphere $S^2$ will  be studied in Subsection \ref{subsect:imageonPoissonsphere} to explain the connection between the topology of the $ S _k $ and the motion of the rigid body in physical space.

 \subsection{Projection onto the torus} \label{subsect:angularparam}
 In order to describe the projection of the level surfaces $ S _k $ onto the $2$-torus $T^2_k$
it is convenient to use a standard parametrization and describe the torus as the square flat torus.  
 Since  the first equation in \eqref{eq:suslovsurface} defines an ellipse in the $(m_1, \gamma_1)$-plane and  the second equation defines an ellipse in the   $(m_2, \gamma_2)$-plane, we parametrize these ellipses by the angle in the respective polar coordinates:
  \begin{equation}\label{eqn:parametrization}  \left\{
  \begin{matrix}
   m_1  = & \sqrt{k_1}\cos \theta_1\\
   \gamma_1  = & \sqrt{\frac{k_1}{b_1}} \sin \theta_1
  \end{matrix}
  \right. \text{ and }
  \left\{
  \begin{matrix}
   m_2 = & \sqrt{k_2}\sin \theta_2\\
   \gamma_2  = & \sqrt{\frac{k_2}{b_2}} \cos \theta_2
  \end{matrix}
  \right.
  ,~\theta_1 \in \left[-\frac{\pi}{2}, \frac{3\pi}{2}\right) \text{ and } \theta_2 \in \left[0, 2\pi\right)
  \end{equation} 
Here the square $ \left[ - \frac{ \pi } { 2 } ,  \frac{3\pi}{2}\right] \times  \in \left[0, 2\pi\right] $  is  viewed  as the square flat torus $ T ^2 $ by   identifying the top side of the square with the bottom side, and the left side with the right side. 

  Then the parametrization above  defines the following   isomorphism from $T^2_k$ to the standard torus $T^2$:
  \[\varphi_k: T^2_k \xto{\cong} T^2 : (m_1, m_2, \gamma_1, \gamma_2)  \mapsto  (\theta_1, \theta_2)\]
  By definition, for each $k$, we see that $S_k \subset T^2_k \times \RR$, where the coordinate on the second factor is given by $\gamma_3$. Let $p_k : S_k \to T^2_k$ be the projection induced by the projection of $T^2_k \times \RR$ to the first factor.
  The dependence of $S_k$ on $k$ can be described using $\varphi_k\circ p_k$.

 To describe the projection of the surfaces $ S _k $ onto the torus (or more precisely the image of $ S _k $ under the map $ \phi _k \circ p _k $) it is convenient to introduce the function $ g_k : T ^2 \to \mathbb{R}  $ defined as  
  \[g_k(\theta_1, \theta_2) = \displaystyle{\frac{k_1}{b_1} \cos^2\theta_1 + \frac{k_2}{b_2} \sin^2\theta_2}\]
  and let $k = (k_1, k_2)$.  We denote by $ U_k $ the subset of $ T ^2 $ consisting of all points at which $ g_k $ takes values greater than the real number $ \epsilon_k $, that is we set   
  \[U_k = \left\{(\theta_1, \theta_2)\in T ^2  : g_k  > \varepsilon_k = \frac{k_1}{b_1} + \frac{k_2}{b_2} - 1\right\} \subseteq T^2\] 
  and let $\partial U_k = g_k^{-1}(\varepsilon_k)$. Let $\bar U_k = U_k \union \partial U_k$ be the closure of $U_k$. Figure \ref{fig:U} 
 shows the set $ \bar U _k $ for various values of $k$.

\begin{figure}[h!]
    \centering
    \begin{subfigure}[b]{0.3\textwidth}
        \includegraphics[width=\textwidth]{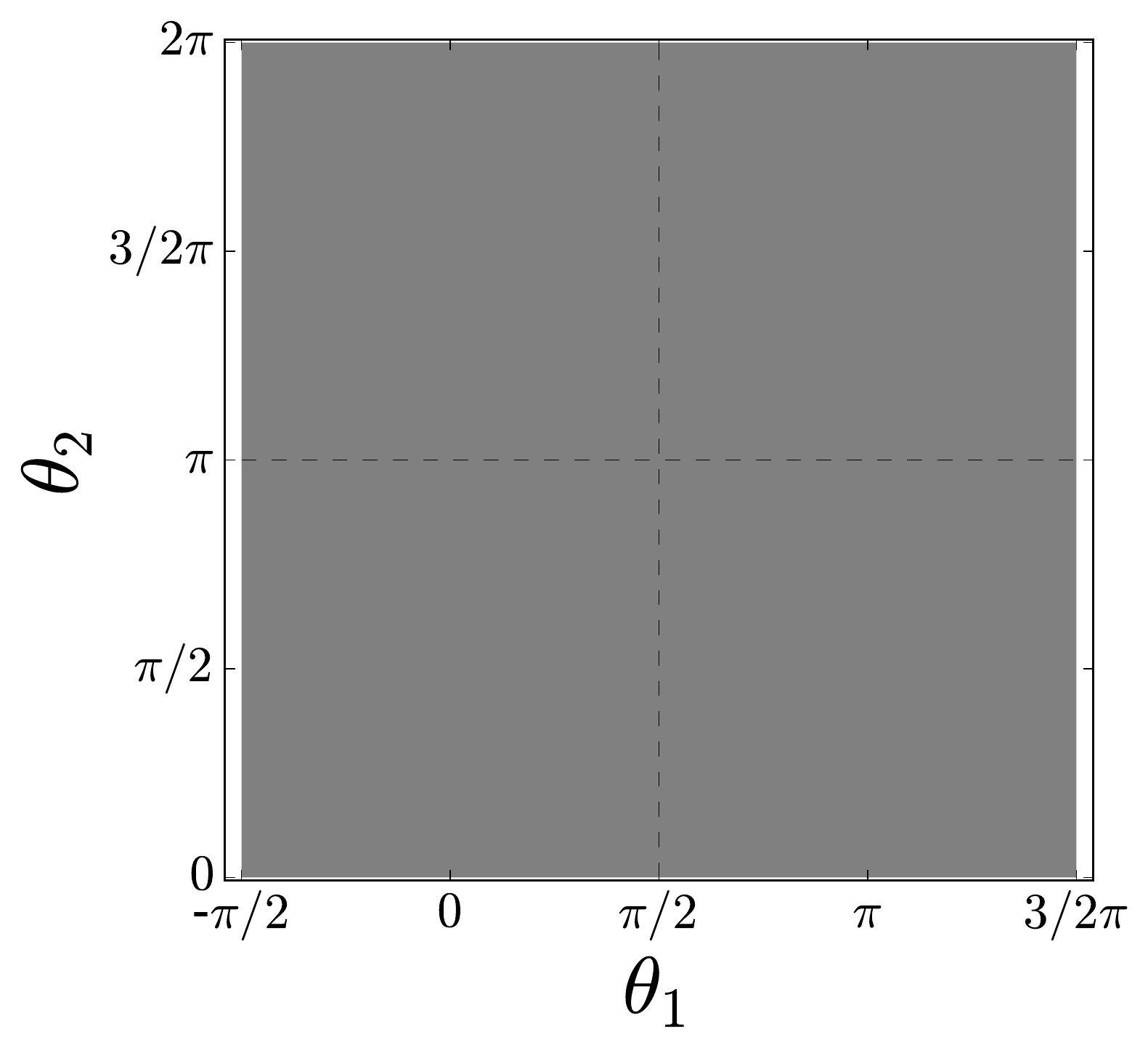}
        \caption{ $ \bar U_k $ for $ k \in D_1 $}
        \label{fig:D1}
    \end{subfigure}
    ~ 
    \begin{subfigure}[b]{0.3\textwidth}
        \includegraphics[width=\textwidth]{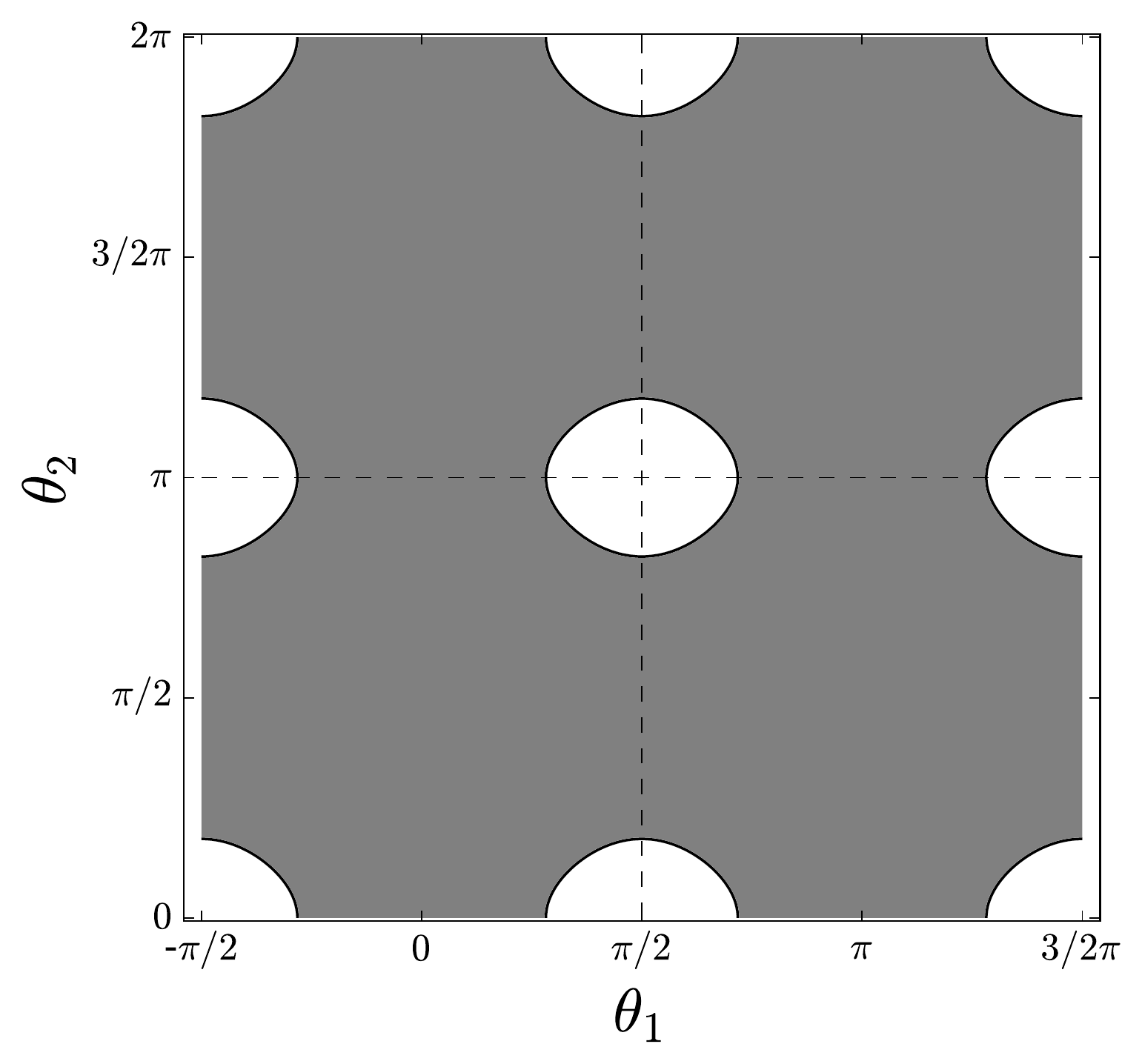}
        \caption{ $ \bar U_k $ for $ k \in D_2 $ }
        \label{fig:D2}
    \end{subfigure}
    ~ 
    \begin{subfigure}[b]{0.3\textwidth}
        \includegraphics[width=\textwidth]{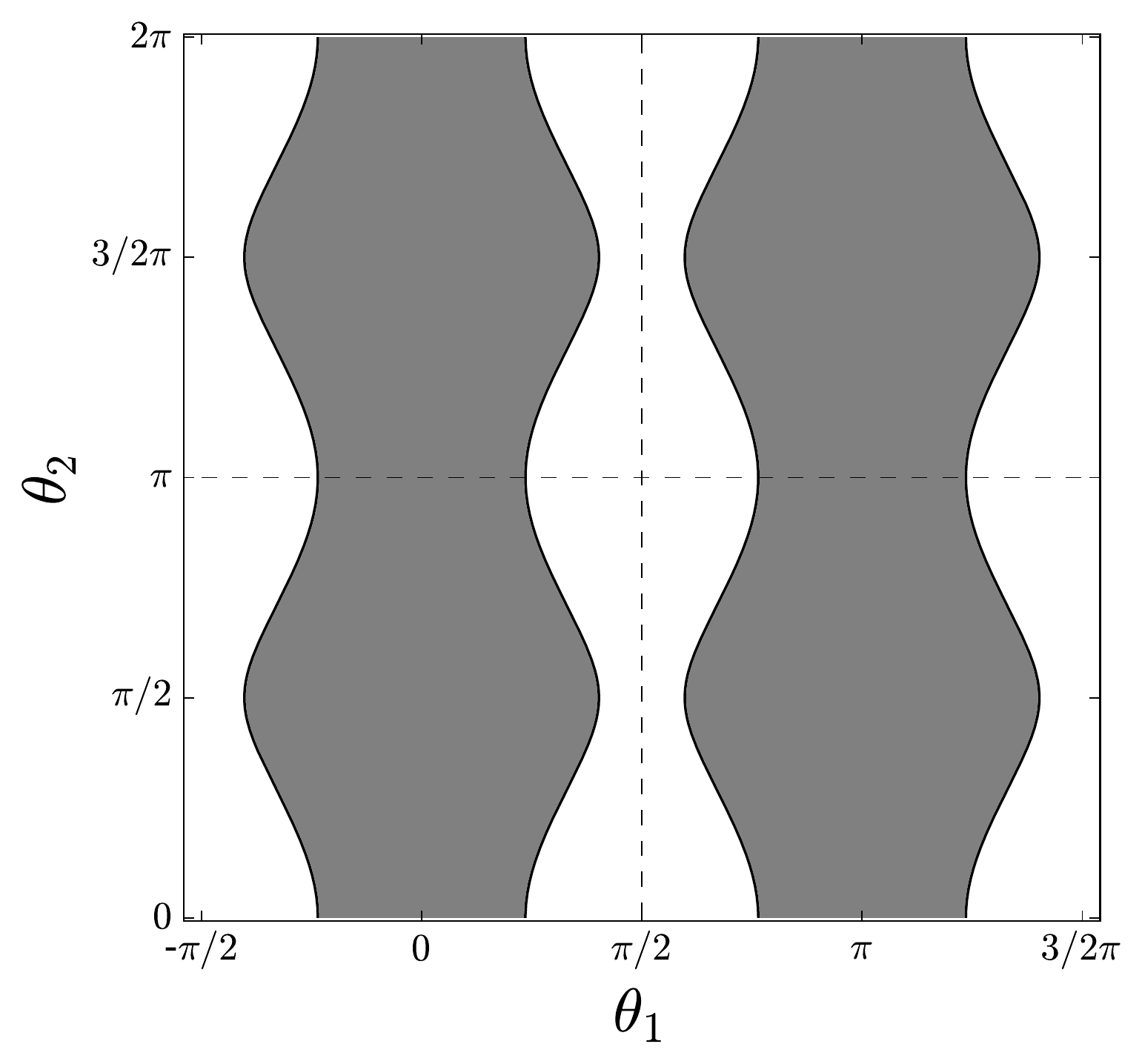}
        \caption{$ \bar U_k $ for $ k \in D_3 $ }
        \label{fig:D3}
    \end{subfigure}
        \begin{subfigure}[b]{0.3\textwidth}
        \includegraphics[width=\textwidth]{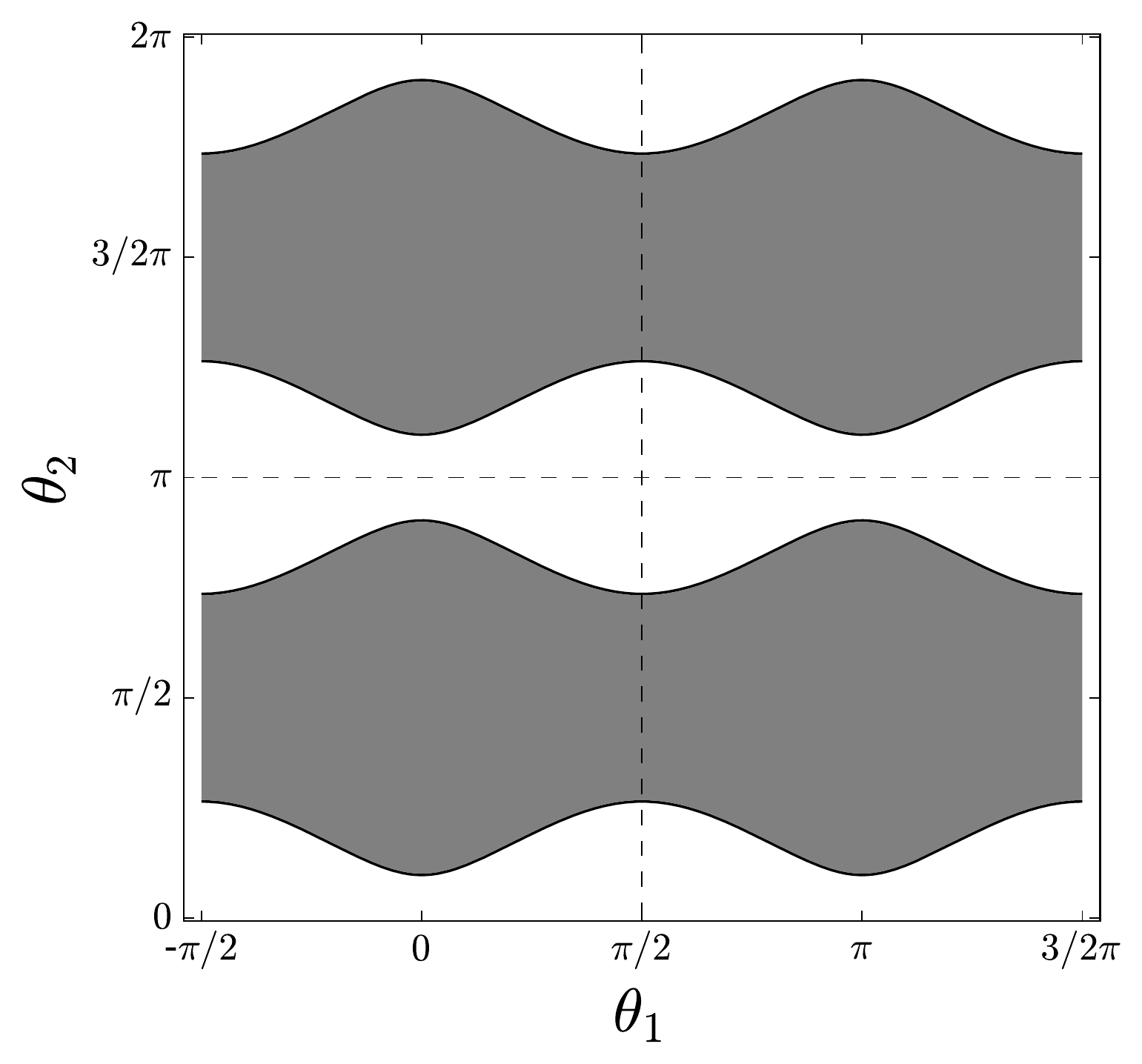}
        \caption{ $ \bar U_k $ for $ k \in D_4 $}
        \label{fig:D4}
    \end{subfigure}\hspace {1 cm} 
     \begin{subfigure}[b]{0.3\textwidth}
        \includegraphics[width=\textwidth]{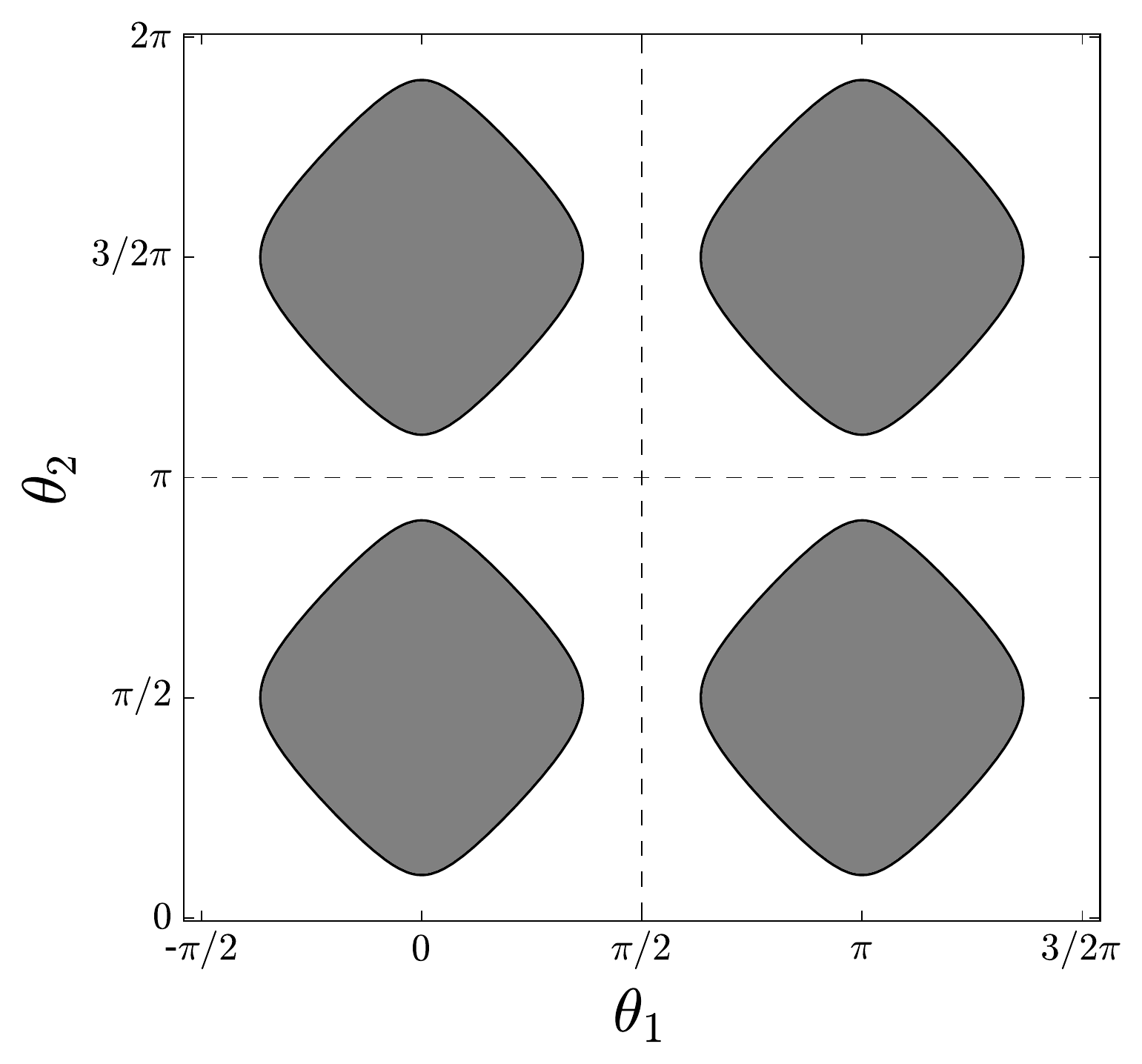}
        \caption{$ \bar U_k $ for $ k \in D_5 $}
        \label{fig:D5}
    \end{subfigure}
    \caption{The set $ \bar U_k $ for various values of $ k $.}\label{fig:U}
\end{figure}

  The following Lemma shows that the set $ \bar U_k $ is the image of $ S_k $ under the map $ \varphi_k \circ p_k $ and gives a  characterization of  such image.  
  \begin{lemma}\label{lemma:projectionimage}
   For all $k$, we have $\bar U_k = \varphi_k\circ p_k(S_k)$. The map $\varphi_k \circ p_k : S_k \to \bar U_k$ is $2$-to-$1$ over the interior $U_k$ and is $1$-to-$1$ over the boundary $\partial U_k$, if $\partial U_k \neq \emptyset$. Thus, $S_k$ is homomorphic to the surface obtained by attaching two copies of $\bar U_k$ along $\partial U_k$.
  \end{lemma}
 \begin{proof}
   The equation defining $S_k$ in $T^2_k \times \RR$ is $\gamma_1^2 + \gamma_2^2 + \gamma_3^2 = 1$. It follows that $\varphi_k\circ p_k(S_k)$ is the subset of $T^2$ defined by
   \[1 \vargeq \gamma_1^2 + \gamma_2^2 = \frac{k_1}{b_1}\sin^2\theta_1 + \frac{k_2}{b_2} \cos^2\theta_2\]
   This is exactly $\bar U_k$. Over $U_k$, the strict inequality above holds, which implies that $\gamma_3 \neq 0$ takes $2$ distinct values. It follows that $\varphi_k\circ p_k$ is $2$-to-$1$ over $U_k$. Over the boundary $\partial U_k$, the equality holds and it implies that $\gamma_3 = 0$. Thus $\varphi_k\circ p_k$ is $1$-to-$1$ along $\partial U_k$ whenever it is not empty.
  \end{proof}
 A consequence of this result is that we can use  the shape of the set  $ \bar U_k $ to characterize the geometry and the topology of the surfaces $ S_k $. The following Lemma describes some feature  of the function $ g_k $ that can be used to describe the shape of the sets $ \bar U_k $. 
 \begin{lemma}\label{lemma:glevels}
  The smooth function $g_k$ on $T^2$ is a Morse function. The $16$ critical points are independent of $k$, with $4$ critical points on each of the $4$ critical levels:
    \begin{enumerate}
   \item Minimums at $\displaystyle{\left\{-\frac{\pi}{2}, \frac{\pi}{2}\right\} \times \left\{0, \pi\right\} }$, with $g_k = 0$.
   \item Saddles at $\displaystyle{\left\{-\frac{\pi}{2}, \frac{\pi}{2}\right\} \times \left\{\frac{\pi}{2}, \frac{3\pi}{2}\right\}}$, with $g_k  = \displaystyle{\frac{k_1}{b_1}}$.
   \item Saddles at $\displaystyle{\left\{0, \pi\right\} \times \left\{0,\pi\right\}}$, with $g_k  = \displaystyle{\frac{k_2}{b_2}}$.
   \item Maximums at $\displaystyle{\left\{0, \pi\right\} \times \left\{\frac{\pi}{2}, \frac{3\pi}{2}\right\}}$, with $g_k  = \displaystyle{\frac{k_1}{b_1} + \frac{k_2}{b_2}}$.
  \end{enumerate}
 \end{lemma}
\begin{proof}
    The statement follows from straightforward computations. 
\end{proof}

 We can now use   Lemmata \ref{lemma:projectionimage} and \ref{lemma:glevels} to completely characterize the surfaces $S_k$.
 \begin{prop}\label{prop:dim3components}  
The  topology of the  surfaces $S_k$ is described below: 
\begin{itemize}
   \item For $k \in D_1$, $S_k$ is isomorphic to two copies of $T^2$.
   \item For $k \in D_2$, $S_k$ is a genus $5$ surface.
   \item For $k \in D_3$ or $D_4$, $S_k$ is isomorphic to two copies of $T^2$.
   \item For $k \in D_5$, $S_k$ is isomorphic to four copies of $S^2$.
  \end{itemize}
 \end{prop}
 \begin{proof}
 The results follow from understanding the set $U_k$ using Lemma \ref{lemma:glevels}.
 For $ k \in D_1$ we have $\displaystyle{\frac{k_1}{b_1} + \frac{k_2}{b_2} < 1}$, which gives $\varepsilon_k < 0$. Since $g_k  \vargeq 0$, we see that $\partial U_k = \emptyset$, see Figure \ref{fig:D1}. Thus $S_k$ is isomorphic to two copies of $T^2$.

%
  For $k \in D_2$, we have
  \[k_1 < b_1, k_2 < b_2, \frac{k_1}{b_1} + \frac{k_2}{b_2} > 1 \Longrightarrow 0 < \varepsilon_k < \min\left\{\frac{k_1}{b_1}, \frac{k_2}{b_2}\right\}\]
  It follows by Lemma \ref{lemma:glevels} that the set $U_k$ is isomorphic to $T^2 \setminus 4D^2$, and $\partial U_k \cong 4S^1$, see Figure \ref{fig:D2}. Lemma \ref{lemma:projectionimage} implies that $S_k$ is isomorphic to two copies of $T^2$ connect sum at $4$ distinct points, i.e. a genus $5$ surface.

  For $k \in D_3$, we have
  \[k_1 > b_1, k_2 < b_2 \Longrightarrow \frac{k_2}{b_2} < \varepsilon_k < \frac{k_1}{b_1}\]
  Then Lemma \ref{lemma:glevels} implies that $U_k$ consists of two components $C_{k,1} \union C_{k,2}$, each of which is isomorphic to $S^1 \times (0,1)$, and $\partial U_k \cong 4S^1$, see Figure \ref{fig:D3}. Apply Lemma \ref{lemma:projectionimage}, we see that $S_k$ has two components as well, each of which is isomorphic to a $2$-torus.
  The argument for $k \in D_4$ is similar, where we have
  \[k_1 < b_1, k_2 > b_2 \Longrightarrow      \frac{k_1}{b_1} < \varepsilon_k < \frac{k_2}{b_2}\]
  and $ U_k $ again consists of two components, see Figure \ref{fig:D4}.
  Again, $S_k$ in this case has two components and each is isomorphic to a $2$-torus.

  For $k \in D_5$, we have
  \[k_1 > b_1, k_2 > b_2 \Longrightarrow \max\left\{\frac{k_1}{b_1}, \frac{k_2}{b_2}\right\} < \varepsilon_k < \frac{k_1}{b_1} + \frac{k_2}{b_2}\]
  Lemma \ref{lemma:glevels} implies that $U_k$ consists of $4$ components $H_{k,j}$ for $j = 1, 2, 3, 4$, each of which is isomorphic to $D^2$, the $2$-disk, and $\partial U_k \cong 4 S^1$, see Figure \ref{fig:D5}. With Lemma \ref{lemma:projectionimage}, we find that $S_k$ has four components, each of which is isomorphic to an $S^2$.
 \end{proof}

 \section{Dynamics on the level surfaces \texorpdfstring{$S_k$}{Sk}}\label{sect:dynamics3dim}
 The projection to the torus as described in the previous section also provides us with detailed information on the Suslov flow.
 
 \subsection{Linear flow on tori}\label{subsect:linearflow}
 We consider the region $D_1$ in Figure \ref{fig:bifurcations}, where the level sets are two tori.  From the proof of Proposition \ref{prop:dim3components} each component of the level set at $(k_1, k_2) \in D_1$ is diffeomorphic to the torus in $\RR^4$ with coordinates $(m_1, m_2, \gamma_1, \gamma_2)$ given by
  \[\{m_1^2 + b_1\gamma_1^2 = k_1, m_2^2 + b_2 \gamma_2^2 = k_2\}\]
 Each equation above defines an ellipse in the plane, which, as we have seen,  can be parametrized by introducing polar coordinates in each plane \eqref{eqn:parametrization}.  In these coordinates, on the level surface, the Suslov flow \eqref{eq:suslovflow}  takes the form 
  \[\dot\theta_1 = \sqrt{b_1}\gamma_3,\quad  \dot \theta_2 = \sqrt{b_2} \gamma_3\]
 Thus the Suslov flow projects to a linear flow with slope $\displaystyle{\sqrt{ \frac{ b_2 } { b_1 } } } $ on the torus, which is periodic when the ratio is a rational number. On the square flat torus $ T ^2 $, the projected flow is given by pieces of straight lines with slope $\displaystyle{\sqrt{ \frac{ b_2 } { b_1 } } }$,  see figure \ref{fig:Z1}.

 \begin{figure}[h!]
    \centering
    \begin{subfigure}[b]{0.45\textwidth}
        \includegraphics[width=\textwidth]{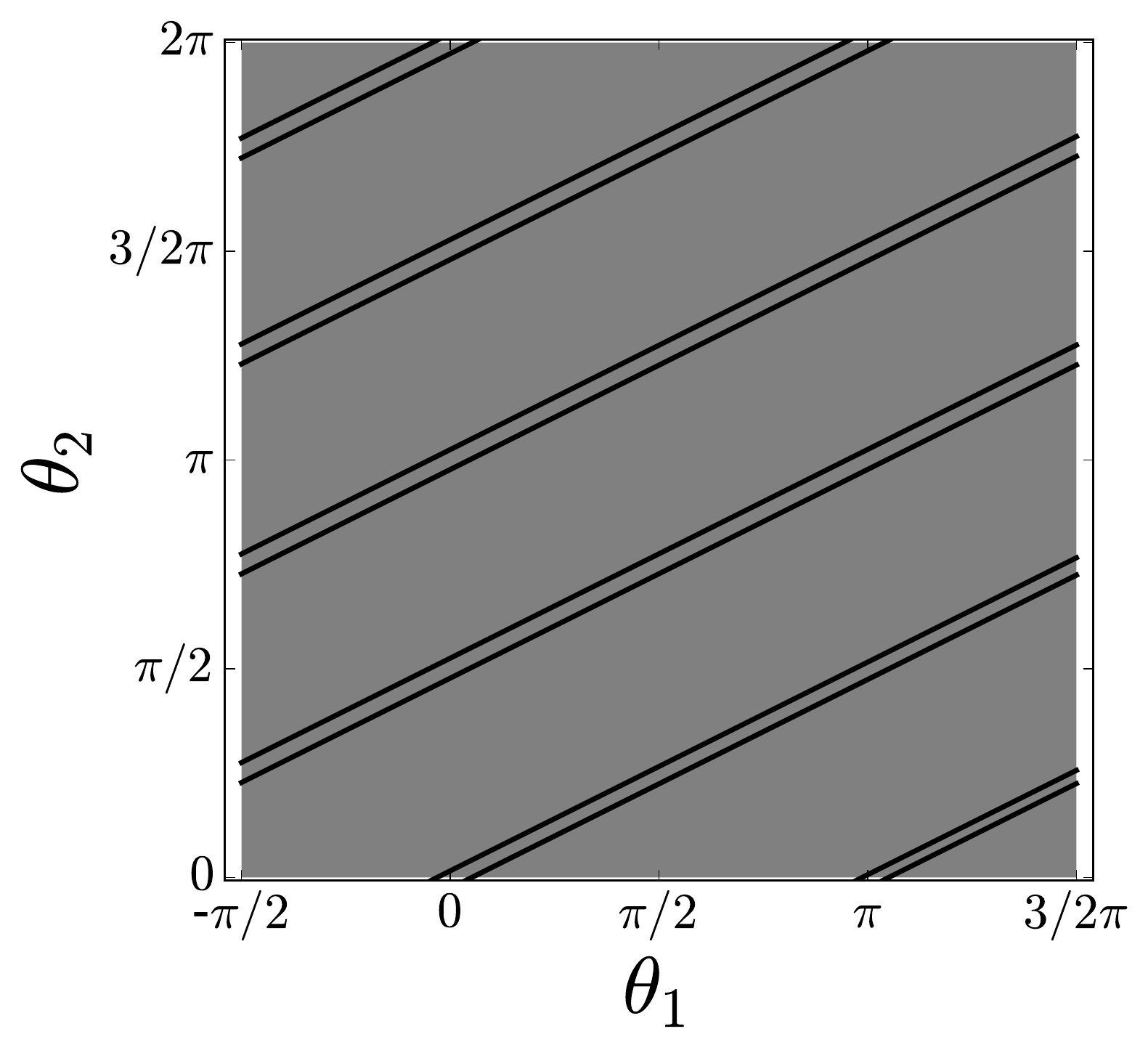}
        \caption{ $ b_1 = 4 $ , $ b_2 = 1 $, $ k_1 =1 $ , and $ k_2 =0.5 $ }
        \label{fig:Z1}
    \end{subfigure}\qquad\qquad
    ~ 
    \begin{subfigure}[b]{0.45\textwidth}
        \includegraphics[width=\textwidth]{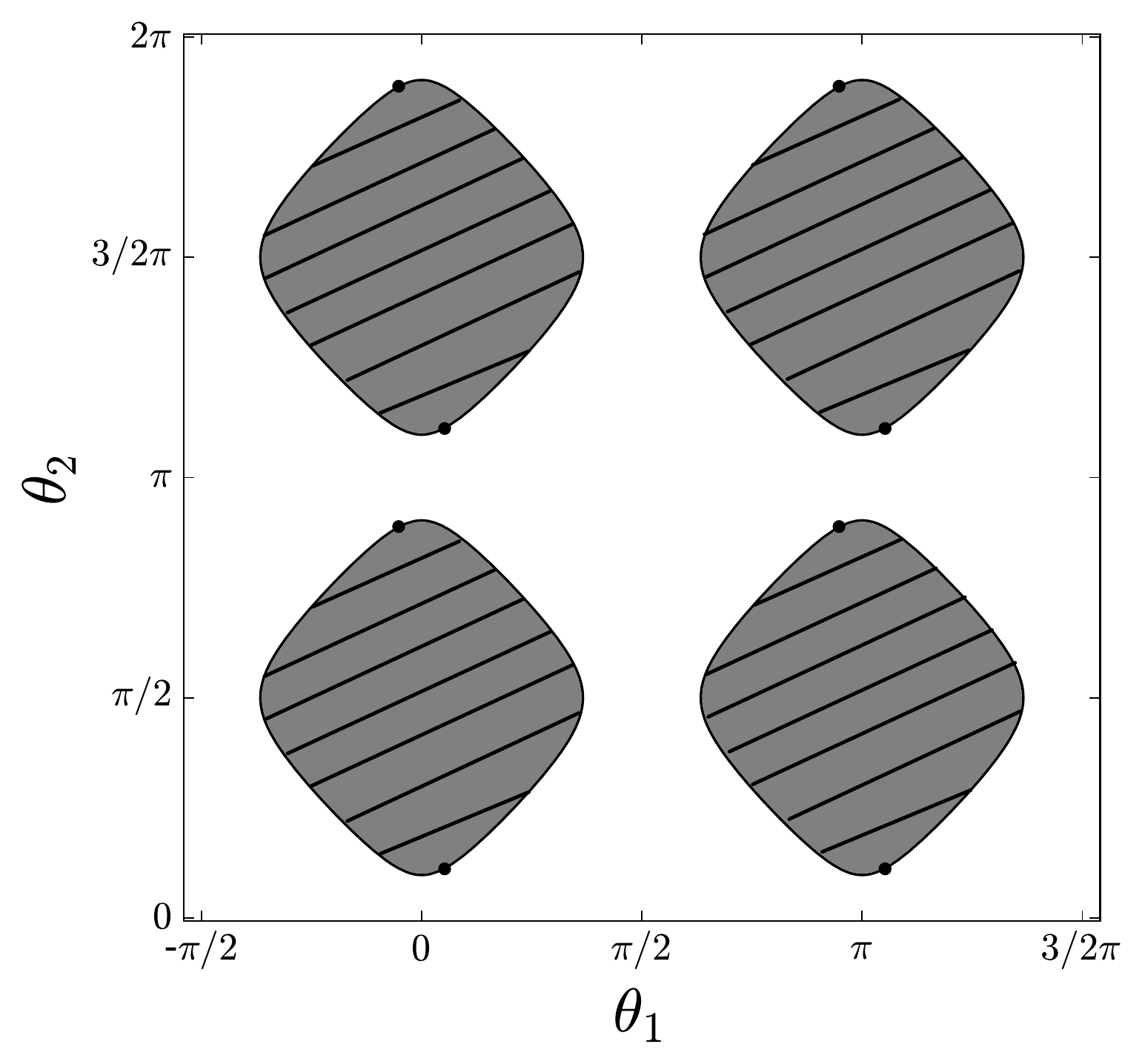}
        \caption{ $ b_1 = 4 $ , $ b_2 = 1 $, $ k_1 =4.4 $ , and $ k_2 =1.1 $   }
        \label{fig:Z5}
    \end{subfigure}

    \caption{(a) The flow for $ k \in  D_1 $. In this case have rational slope $ \sqrt{ b_2 /b_1 } = 1/2 $, and all the orbits are periodic. The picture shows four periodic orbits.  (b)  The flow for $ k \in D_5 $. There are 8 critical points. All the other orbits are periodic. }
\end{figure}

 \subsection{Additional integral of motion} \label{subsect:newintegralofmotion}
  When $\displaystyle{\sqrt{\frac{b_2}{b_1}}} \in \QQ$ the system admits another integral of motion, which implies as well that Suslov flow on $S_k$ is periodic for generic $k$ in this case. We describe it first for $b_1 = b_2 = b$.
 \begin{prop}\label{prop:newintegral}
  When $b_1 = b_2 = b$, the Suslov flow has the following as an integral of motion:
  \[f_3 = m_1m_2 - b\gamma_1 \gamma_2 \]
 \end{prop}
 \begin{proof}
  Straightforward verification by taking derivative with respect to $t$.
\end{proof} 
 It's readily verified that the level sets of $f_3$ define the periodic flow on the tori for $k \in D_1$ when $b_1 = b_2$.
 In general, the new integral of motion is a higher degree polynomial in $m_i$'s and $\gamma_i$'s.
  \begin{prop}\label{prop:rationalratiointegral}
  Suppose that the ratio $\sqrt{b_1} : \sqrt{b_2}$ is rational, then there is an integral of motion $f_3$, given by a polynomial of $(\gamma_1, \gamma_2, m_1, m_2)$.
 \end{prop}
 \begin{proof} 
 For a given $k \in D_1$, rewrite the flow equations \eqref{eq:suslovflow} in the $(\theta_1, \theta_2)$-coordinates:
  \[ \left\{
  \begin{matrix}
   m_1  = & \sqrt{k_1}\cos \theta_1\\
   \gamma_1  = & \sqrt{\frac{k_1}{b_1}} \sin \theta_1
  \end{matrix}
  \right. \text{ and }
  \left\{
  \begin{matrix}
   m_2 = & \sqrt{k_2}\sin \theta_2\\
   \gamma_2  = & \sqrt{\frac{k_2}{b_2}} \cos \theta_2
  \end{matrix}
  \right.
  \]
  where $ - \frac{ \pi } { 2 } \leq \theta_1 \leq \frac{ 3 \pi } { 2 } $ and $ 0 \leq \theta_2 \leq 2 \pi $, and we have
   \[\dot\theta_1 = \sqrt{b_1}\gamma_3, \dot \theta_2 = \sqrt{b_2} \gamma_3\]
  Let $p, q \in \ZZ$ be integers such that
   \[\sqrt{\frac{b_1}{b_2}} = \frac{p}{q}\]
  then we see that $q\theta_1 - p\theta_2$ is a constant along the flow
   \[q \dot \theta_1 - p \dot \theta_2 = (q\sqrt{b_1} - p\sqrt{b_2})\gamma_3 = 0\]
  Furthermore, we can express trigonometric functions of $q\theta_1 - p\theta_2$ as a degree $p+q$ polynomial in $m_1, m_2, \gamma_1, \gamma_2$, involving also $k_1$ and $k_2$. For example, let $z_1 = e^{i\theta_1}$ and $z_2 = e^{i\theta_2}$, then
  \[\cos(q\theta_1 - p\theta_2) = \Re(z_1^q z_2^{-p}) = \Re((\cos\theta_1 + i \sin\theta_1)^q(\cos\theta_2 - i\sin\theta_2)^{p})\]
  It follows that 
  \[f_3 = \Re\left(\left(m_1 + i\sqrt{b_1}\gamma_1\right)^q\left(\sqrt{b_2}\gamma_2 - i m_2\right)^p\right) = k_1^{\frac{q}{2}}k_2^{\frac{p}{2}} \cos(q\theta_1 - p\theta_2)\]
  is a constant along the flows for $k \in D_1$. It is straightforward to verify by direct differentiation that $\left(m_1 + i\sqrt{b_1}\gamma_1\right)^q\left(\sqrt{b_2}\gamma_2 - i m_2\right)^p$ is a constant along the flow independent of $k$. Thus $f_3$ is an integral of motion, which is a degree $p+q$ real polynomial in $(\gamma_1, \gamma_2, m_1, m_2)$.
 \end{proof} 
 \subsection{Critical Points}
Critical points of the flow of the Suslov problem can be obtained by a simple geometric argument. We observe that the critical points are precisely where the level sets $  \partial U_k $ are tangent to the linear flow. Thus, in $(\theta_1, \theta_2)$ coordinates, the critical points are exactly the solutions to the following system of equations:
\begin{align*}
    & \frac{ k_1 } { b_1 } \sin ^2 \theta_1 + \frac{ k_2 } { b_2 } \cos  ^2 \theta_2=1\\
    &  \frac{ k_1 } { b_1 } \frac{ b_2 } { k_2 } \frac{ \sin \theta_1 \cos \theta_1 } { \sin \theta_2 \cos \theta_2 } = \sqrt{ \frac{ b_2 } { b_1 } }
\end{align*}
The second equation simplifies to
\[\frac{k_1^2}{b_1} \sin^2\theta_1 (1-\sin^2\theta_1) = \frac{k_2^2}{b_2}(1-\cos^2\theta_2)\cos^2\theta_2\]
Using the first equation and the fact that $\gamma_1^2 = \displaystyle{\frac{k_1}{b_1}\sin^2\theta_1}$, we obtain the following quadratic equation in $\gamma_1^2$, which can be explicitly solved:
\begin{equation}\label{eqn:quadratic_gamma1} 
    (b_1 - b_2) \gamma_1 ^4 - (k_1 + k_2 - 2 b_2) \gamma_1 ^2 + (k_2 - b_2) = 0.    
\end{equation}

When $ b_1 = b_2 : = b $, \eqref{eqn:quadratic_gamma1} reduces to 
\[(k_1 + k_2 - 2 b ) \gamma_1 ^2 + (k_2 - b ) = 0 \Longrightarrow \gamma_1 = \pm \gamma_1^*, \text{ where } \gamma_1 ^\ast =  \sqrt{ \frac{ k_2 - b } { k_1 + k_2 - 2b } } \]
Let 
\[\gamma_2 ^\ast =  \sqrt{ \frac{ k_1 - b } { k_1 + k_2 - 2b } } \text{ and } t = \sqrt{ k_1 + k_2 - b }\]
then the critical points in this case are given in Table \ref{table:1} below:
 \begin{table}[h!]
    \footnotesize{\begin{tabular}{|l|l|c|l|}
  \hline 
   Region & \multicolumn{1}{c|}{Value of Parameters}&  \# of cp &  critical points: $ (m_1 , m_2 , \gamma_1 , \gamma_2 , \gamma_3 ) $   \\

  \hline\hline 
  $ D_1$         &  $ k_1 + k_2 < b $ & 0 &  \\
$ D_2 $  &    $ k_1 + k_2 >b,\, k_1 <b, \, k_2<b$   & 8 &   $ \pm (\pm t \gamma_2^\ast ,\pm t \gamma_1^\ast,\pm\gamma_1 ^\ast , \pm \gamma_2 ^\ast ,0) $,   \\
 &  &&         $ \pm (\pm t \gamma_2^\ast ,\mp t \gamma_1^\ast,\mp\gamma_1 ^\ast , \pm \gamma_2 ^\ast ,0) $     \\
  $ D_3$        & $  k_1> b, k_2 < b   $   &  0 &  \\
   $ D_4$      &$  k_1< b, k_2 > b   $   & 0 &  \\
 $ D_5$  &    $  k_1 > b, k_2 > b $    & 8  &  $  \pm (\pm t \gamma_2^\ast , \pm t  \gamma_1^\ast,\pm \gamma_1 ^\ast , \pm \gamma_2 ^\ast ,0) $,\\ 
 &     &     &   $ \pm (\pm t \gamma_2^\ast ,\mp t \gamma_1^\ast,\mp\gamma_1 ^\ast , \pm \gamma_2 ^\ast ,0) $   \\
\hline  
\end{tabular}}\caption{Critical points for $ b_1 =  b_2 $. See Figure \ref{fig:bifurcations}  for a definition of the regions $ D_i $.} \label{table:1}
\end{table}

Suppose that $b_1 \neq b_2$. As a quadratic equation in $\gamma_1^2$, the discriminant of \eqref{eqn:quadratic_gamma1} is
\begin{equation}\label{eq:discriminant}
    \Delta = (k_1 + k_2 - 2 b_2) ^2 - 4 (b_1 - b_2) (k_2 - b_2)
\end{equation}
The solutions of equation \eqref{eqn:quadratic_gamma1} are $\gamma_1 = \pm \Gamma_1 ^{ - }, \, \pm \Gamma_1 ^{ + }$ with
\begin{equation}\label{eqn:solutions_quartic} 
     \Gamma_1 ^{ - } =  \sqrt{ \frac{ (k_1 + k_2 - 2 b_2) - \sqrt{ \Delta }}{2 (b_1 - b_2 )  }  }, \, \Gamma_1 ^{ + }=  \sqrt{ \frac{ (k_1 + k_2 - 2 b_2) + \sqrt{ \Delta }}{2 (b_1 - b_2 )  }  }  
\end{equation} 
which can be real or complex depending on the value of the parameters. 
 \begin{figure}[!ht]
   \begin{minipage}{0.4\textwidth}
    \begin{tikzpicture}[scale=2.5]
     \draw[very thick,->] (-0.25,0) -- (2.2,0) node[right] {$k_1$};
     \draw[very thick,->] (0,-0.25) -- (0,2.2) node[above] {$k_2$};

     \fill[yellow,opacity=0.6] (1,0.5) -- (2,0.5) -- (2,2) -- (1,2)--cycle;
     \fill[red,opacity=0.6] (1,0) -- (1,0.5) -- (0,0.5) --cycle;
     \fill[blue!60!cyan,opacity=0.3] (1,0) -- (2,0) -- (2,0.5) -- (1,0.5)--cycle;
     \fill[blue!60!cyan,opacity=0.3] (0,0.5) -- (0,2) -- (1,2) -- (1,0.5)--cycle;
     \fill[blue!60!cyan,opacity=0.3] (0,0) -- (0,0.5) -- (1,0) --cycle;

     \draw[thick] (1,0) -- (0,0.5);
     \draw[thick] (1,0) --  (1,2);
     \draw[thick] (0,0.5) --  (2,0.5);
     \node at (1.5,1.2) {\small$D_5$ };
     \node at (0.7,0.35) {\small$D_2$ };
     \node at (0.3,0.15) {\small$D_1$ };
     \node at (-0.2,0.6) {\small$D_4^1 $ };
     \draw[->] (-0.1,0.57) --  (0.1,0.53); 
     \node at (1.5,0.3) {\small$D_3$ };
     \node at (0.5,1.2) {\small$D_4^4 $ };
     \node at (0.5,0.9) {\small$\Delta <0 $ };
     \node at (0.9,1.85) {\small$D_4^3 $ }; 
     \node at (0.93,0.57) {\tiny$D_4^2 $ };
     \node at (0.2,1.9) {\small$L $ };
     \node at (1,-0.15) {$ b_1   $ };
     \filldraw (0.5,0.5) circle (0.7pt);
     \filldraw (1,1) circle (0.7pt);
     \draw[thick, scale=0.5,domain=1.46:2,smooth,variable=\x,black,samples=500] plot ({\x},{-\x+2*sqrt(-\x+2)+4});
     \draw[thick,scale=0.5,domain=0:2,smooth,variable=\x,black,samples=500] plot ({\x},{-\x-2*sqrt(-\x+2)+4});
      \draw[dashed,scale=0.5,domain=0:4,smooth,variable=\x,black,samples=500] plot ({\x},{4-\x});
    \end{tikzpicture}
   \end{minipage}
  \begin{minipage}{0.55\textwidth}
    \caption{\label{fig:quartic} 
    The curve $ \Delta = 0 $ divides the region $ D_4 $ into four subregions, namely, $ D_4 ^1 $, $ D_4 ^2 $, $ D_4 ^3 $, and $ D_4 ^4 $, where $D_4^4$ is the subregion where $ \Delta <0 $. The equation for the dashed line $ L $ is $ k_1 + k_2 = 2 b_1 $.  We denote with $ C_1 $ the part of $ \Delta = 0 $ below the line $ L $ and with $ C_2 $ the part above $ L $. }
  \end{minipage}
 \end{figure}
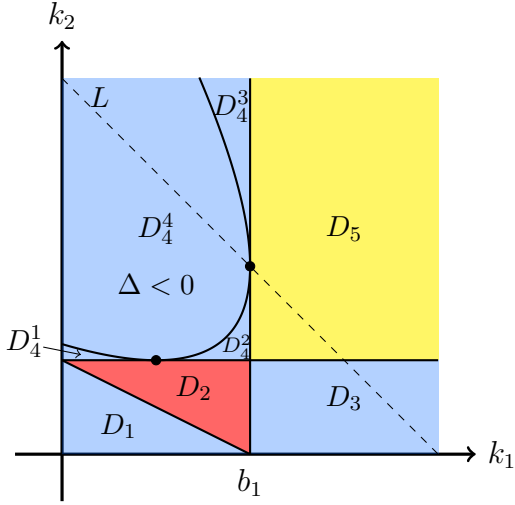
 
\noindent
It's straightforward to see, from \eqref{eqn:solutions_quartic}, that there are no critical points in $D_4^4$.
Let
\[\Gamma_2^\pm = \sqrt{1-\left(\Gamma_1^\pm\right)^2}, \Lambda_1^\pm = \sqrt{k_1 - b_1 (\Gamma_1 ^{ \pm })  ^2} \text{ and } \Lambda_2  ^{\pm} = \sqrt{k_2 - b_2 (\Gamma_2 ^{\pm})  ^2  }\]
%
%
%
%
We can then describe the critical points for $ b_1 > b_2 $ in Table \ref{table:2} below, where the regions are as labeled in Figure \ref{fig:quartic}.

\begin{table}[h!]
    \footnotesize{\begin{tabular}{|l|l|l|l|c|l|}
  \hline 
   Region & \multicolumn{3}{c|}{Value of Parameters}&  \#  &  critical points: $ (m_1 , m_2 , \gamma_1 , \gamma_2 , \gamma_3 ) $   \\

  \hline\hline 
 $ D_1 $  &   &   &  $ k_1 > b_1 $  &  0    &   \\
 $ D_2  $ &  $ k_2 < b_2 $  &   $  \Delta >0 $ & $ k_1 / b_1 + k_2 / b_2 >1$, $ k_1 < b_1 $    &8  &$ \left( (- 1) ^{ k } \Lambda_1 ^{ + }, (- 1) ^{ l } \Lambda_2 ^{ +} , (- 1) ^{ i } \Gamma_1 ^{ + } , (- 1) ^{ j } \Gamma_2 ^{ + },  0\right) $  \\
$ D_3  $ &  &  & $ k_1 / b_1 + k_2 / b_2 <1$    &  0 &    \\
  \hline 
$ D_4 ^3 $   &  &   & $ k_1 < b_1 $,\, $k_1 + k_2 > 2 b_1 $     & 0 &     \\
 $ D_5 $   &   &  $  \Delta >0 $ & $ k_1 > b_1 $    & 8  & $ \left( (- 1) ^{ k } \Lambda_1 ^{ - }, (- 1) ^{ l } \Lambda_2 ^{ - }, (- 1) ^{ i } \Gamma_1 ^{ - } , (- 1) ^{ j } \Gamma_2 ^{ - },  0 \right) $     \\
 $ D_4 ^{ 1 } \cup D_4 ^2 $  &  $ k_2 > b_2 $   &   &  $ k_1 < b_1 $,\, $k_1 + k_2 < 2 b_1 $   & 16  &  $  \left((- 1) ^{ k } \Lambda_1 ^{ \pm }, (- 1) ^{ l } \Lambda_2 ^{ \pm } , (- 1) ^{ i } \Gamma_1 ^{ \pm } , (- 1) ^{ j } \Gamma_2 ^{ \pm },   0\right)$ \\ 
  \cline{3-6} 
$ C_1 $    &   & $ \Delta =0$  & $  k_1 + k_2 < 2 b_1 $  & 8 & $ \left(    (- 1) ^{ k } \Lambda_1 , (- 1) ^{ l } \Lambda_2 ,(- 1) ^{ i } \Gamma_1  , (- 1) ^{ j } \Gamma_2 ,0\right) $   \\
 $ C_2 $     &   & $ \Delta = 0 $ & $ k_1 + k_2 > 2b_1 $  &  0 &  \\
 \cline{3-6}
 $D_4 ^4 $  &  &  $ \Delta <0  $  &  &   $ 0 $   &    \\
 \hline
\end{tabular}}\caption{Critical points for $ b_1 >  b_2 $. Here $ i,j,k \in \{ 0,1\}$ and $ l = i + k - j $.  }\label{table:2}
\end{table}

\subsection{Classification  of Critical points}
Given the explicit computation of all the critical points on the smooth level surfaces, we can now classify all of them.
Recall that the level surface $S_k$ is defined by \eqref{eq:suslovsurface}. The tangent plane at $p \in S_k$ is the kernel of the matrix
\eqref{eqn:gradients} formed by the gradients of the defining equations. Let $p = (m_1, m_2, \gamma_1, \gamma_2, \gamma_3)$ be a critical point of the flow, then $\gamma_3 = 0$ and none of the other coordinates vanishes. Thus, near a critical point $p$ we have a local frame of the tangent space $TS_k$ given by
\[ v_1 = \frac{b_1\gamma_3}{m_1}\frac{\partial}{\partial m_1} -\frac{\gamma_3}{\gamma_1}\frac{\partial}{\partial \gamma_1} + \frac{\partial}{\partial \gamma_3} \text{ and }
   v_2 = \frac{b_1\gamma_2}{m_1}\frac{\partial}{\partial m_1} - \frac{b_2\gamma_2}{m_2} \frac{\partial}{\partial m_2} - \frac{\gamma_2}{\gamma_1}\frac{\partial}{\partial \gamma_1} + \frac{\partial}{\partial \gamma_2} 
\]
We see that integral curves of $v_1$ are given by $\gamma_2 \equiv \const$ and the integral curves of $v_2$ are given by $\gamma_3 \equiv \const$. In particular, $\{\gamma_2, \gamma_3\}$ defines a local coordinate chart around $p$. By an abuse of notation, we may write
\[v_1 = \partial_{\gamma_3} \text{ and } v_2 = \partial_{\gamma_2}\]
Then the Suslov vector field on $S_k$ near $p$ can be written as
\[X = (\gamma_2m_2 - \gamma_1m_1)\partial_{\gamma_3} - m_2\gamma_3 \partial_{\gamma_2}\]
Let $P$ be a critical point of $X$ and suppose that $\gamma_2(P) = c$ and $\gamma_3(P) = 0$. From the equations \eqref{eq:suslovsurface}, we compute that the linearization of $X$ at $P$ to be
\begin{equation}\label{eq:linearize}
 \left(m_2(P) - \frac{b_2\gamma_2(P)^2}{m_2(P)} + \frac{m_1(P) \gamma_2(P)}{\gamma_1(P)} - \frac{b_1\gamma_1(P)\gamma_2(P)}{m_1(P)}\right) (\gamma_2 - c) \partial_{\gamma_3} - m_2(P) \gamma_3 \partial_{\gamma_2}
\end{equation}
which gives the Jacobian of $X$ at $P$:
\[J_X(P) = \begin{pmatrix}
   0 & -m_2(P) \\  m_2(P) - \frac{b_2\gamma_2(P)^2}{m_2(P)} + \frac{m_1(P) \gamma_2(P)}{\gamma_1(P)} - \frac{b_1\gamma_1(P)\gamma_2(P)}{m_1(P)} & 0
  \end{pmatrix}
\]
The characteristic polynomial of $J_X(P)$ is
\[\lambda^2 - m_2(P)\left(m_2(P) - \frac{b_2\gamma_2(P)^2}{m_2(P)} + \frac{m_1(P) \gamma_2(P)}{\gamma_1(P)} - \frac{b_1\gamma_1(P)\gamma_2(P)}{m_1(P)}\right)\]
which simplifies to
\begin{equation}\label{eq:character}
 \lambda^2 + k_1 + k_2 - 2(m_1(P)^2 + m_2(P)^2) 
\end{equation}
since at $P$, we have $\gamma_2(P) m_2(P) - \gamma_1(P) m_1(P) = 0$. The type of the singularity is determined by the roots of the characteristic polynomial in \eqref{eq:character}.

First consider the case where $b_1 = b_2 = b$. In this case, the flow has $8$ critical points on the level set $S_k$ when $(k_1, k_2) \in D_2 \union D_5$, and no critical points in other regions.
\begin{prop}\label{prop:equalbclassify}
 When $b_1 = b_2 = b$, the critical points on $S_k$ are all saddles if $k \in D_2$, and are all centers if $k \in D_5$.
\end{prop}
\begin{proof}
In this case, the explicit coordinates for the singular points in Table \ref{table:1} lead to
 \[m_1(P)^2 + m_2(P)^2 = k_1+k_2-b\]
 which implies that the roots of the characteristic polynomial is given by
 \[\pm\sqrt{2b - (k_1 + k_2)}\]
 The statement follows noticing that $k_1 + k_2 < 2b$ in $D_2$, while $k_1 + k_2 > 2b$ in $D_5$.
\end{proof}

Next, consider $b_1 \neq b_2$. Without loss of generality, we suppose that $b_1 > b_2$.

\begin{proposition}\label{prop:crit_points}
Suppose that $b_1 > b_2$. When $S_k$ is a smooth $2$-manifold, we have:
 \begin{itemize}
   \item If $ (k_1 , k_2) \in D_2 $ then the $ 8 $ critical points are all saddles.
   \item If $ (k_1 , k_2) \in D_4 ^{ 1 } \cup D_4 ^{ 2 }  $ then there are $ 8 $ centers and $ 8 $ saddles.
   \item If $ (k_1 , k_2) \in D_5 $ then the $ 8 $ critical points are all centers.
   \item If $ (k_1 , k_2) \in C_1  $ then there are  $ 8 $ non-hyperbolic critical points.
 \end{itemize}   
\end{proposition}
\begin{proof} 
Using \eqref{eq:suslovsurface} and the fact that $\gamma_3(P) = 0$ at critical point $P$, we see that
\[m_1(P)^2 + m_2(P)^2 = k_1+k_2-b_2 -(b_1- b_2) \gamma_1(P)^2\]
Thus \eqref{eq:character} becomes
\[    \lambda^2 - 2 (b_1 - b_2) \left[\frac{k_1 + k_2 - 2 b_2}{2 (b_1 - b_2)} - \gamma_1(P)^2\right]   \]
When $\Delta \neq 0$, the critical points are non-degenerate. 
Let's call the critical points of the form $ \left( (- 1) ^{ i } \Gamma_1 ^{ + } , (- 1) ^{ j } \Gamma_2 ^{ + },   (- 1) ^{ k } \Lambda_1 ^{ + }, (- 1) ^{ l } \Lambda_2 ^{ + } \right) $ the \emph{$+$-critical points}, and  the critical points of the form  $ \left( (- 1) ^{ i } \Gamma_1 ^{ - } , (- 1) ^{ j } \Gamma_2 ^{ - },   (- 1) ^{ k } \Lambda_1 ^{ - }, (- 1) ^{ l } \Lambda_2 ^{ - } \right) $ the  \emph{$-$-critical points}. At the $\pm$-critical points, by \eqref{eqn:solutions_quartic}, \eqref{eq:character} further simplifies to
\[\lambda^2 \mp \sqrt{\Delta}, \text{ respectively }\]
In particular, all $+$-critical points are saddles and $-$-critical points are centers.
This gives the first three statements.

When $(k_1, k_2) \in C_1$, we have $ \Delta = 0 $ at the critical points and they are all degenerate. The linearization \eqref{eq:linearize} of $X$ at a critical point $P$ here becomes
\begin{equation}\label{eqn:Delta}
    -m_2(P)\gamma_3 \partial_{\gamma_2} \text{ with } m_2(P) \neq 0 \end{equation} 
which implies that they are nonhyperbolic. 
\end{proof}

\subsection{Periodic orbits} \label{subsect:periodorbits}
Recall that on each level surface $S_k$, an orbit of the Suslov flow projects to a portion of an orbit of a linear flow on the torus and the critical points of the Suslov flow correspond to precisely the points where $\partial U_k$ is tangent to the linear flow. Thus a generic orbit of the flow does not contain any critical point in its closure, and we say such generic orbits \emph{non-critical}.

 \begin{figure}[h!]
    \centering
    \begin{subfigure}[b]{0.45\textwidth}
        \includegraphics[width=\textwidth]{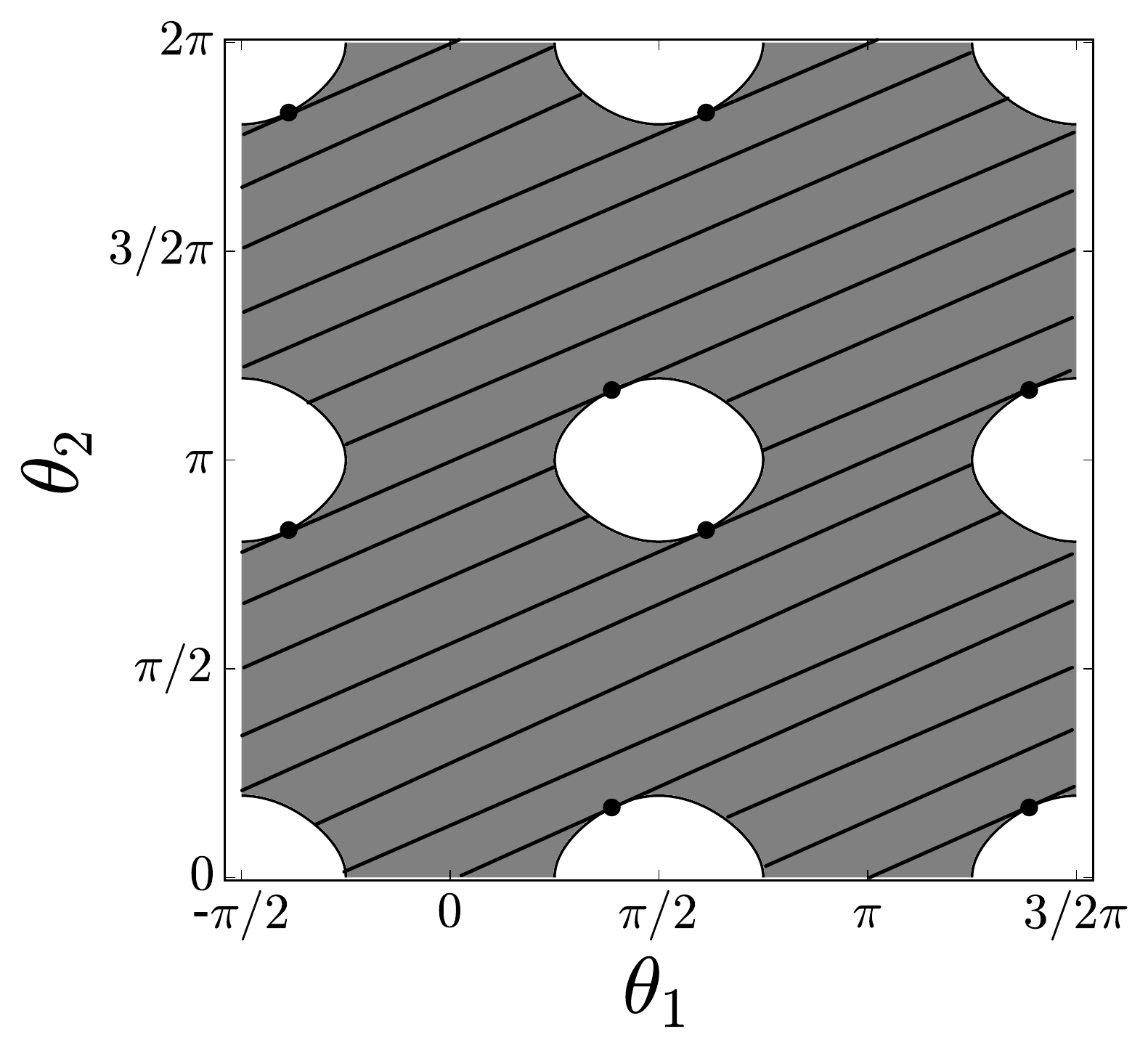}
       \caption{ $ b_1 = 4 $ , $ b_2 = 1 $, $ k_1 =2 $ , and $ k_2 =3/4 $ }
        \label{fig:Z2a}
    \end{subfigure}\qquad\qquad
    ~ 
    \begin{subfigure}[b]{0.45\textwidth}
        \includegraphics[width=\textwidth]{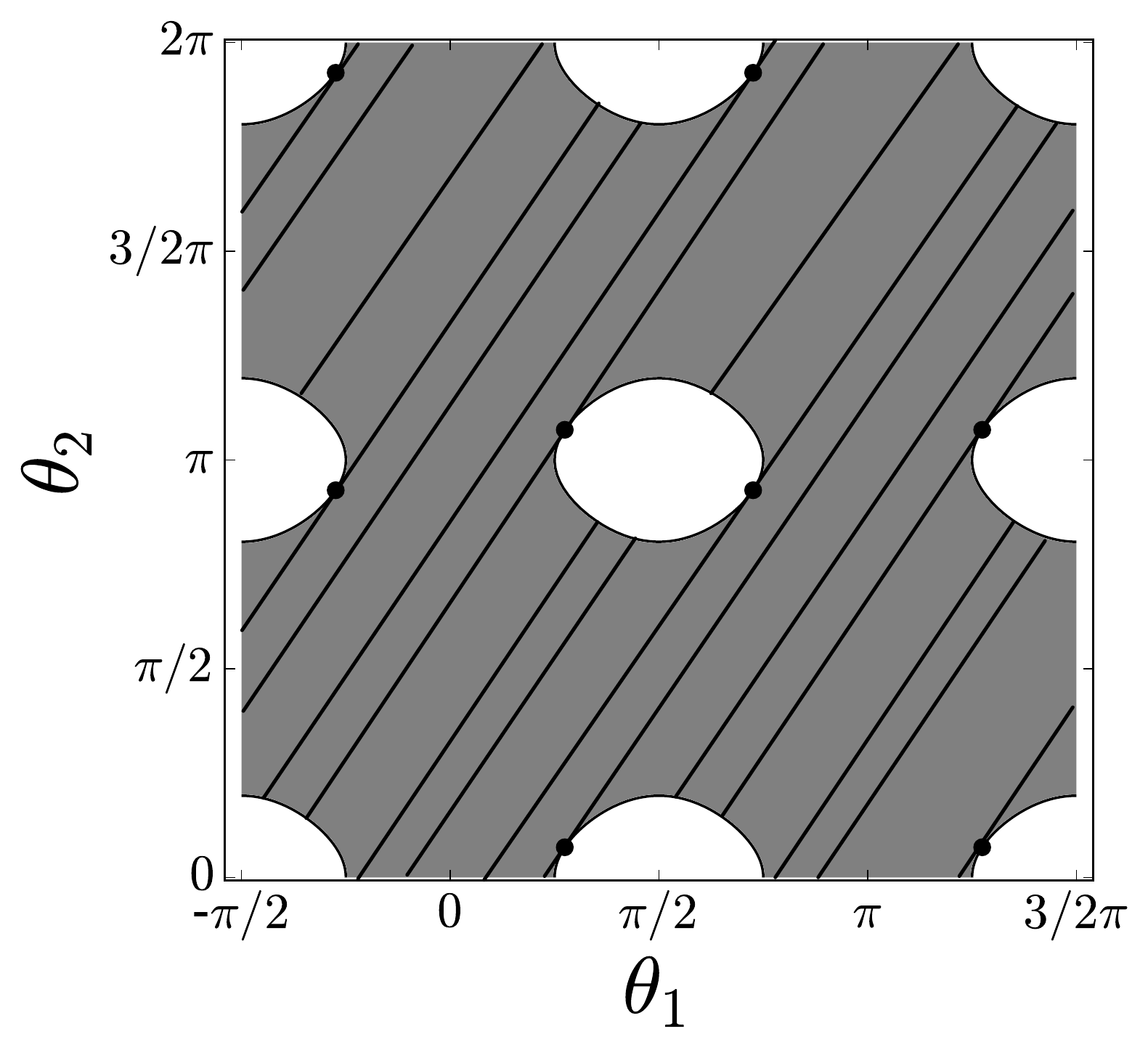}
        \caption{ $ b_1 = 4 $ , $ b_2 = 9 $, $ k_1 =2 $ , and $ k_2 =27/4 $   }
        \label{fig:Z2b}
    \end{subfigure}

    \caption{The flow for $ k \in  D_2$, for two different values of $ k $. 
    }\label{fig:Z2}
\end{figure}
 \begin{figure}[h!]
    \centering
    \begin{subfigure}[b]{0.45\textwidth}
        \includegraphics[width=\textwidth]{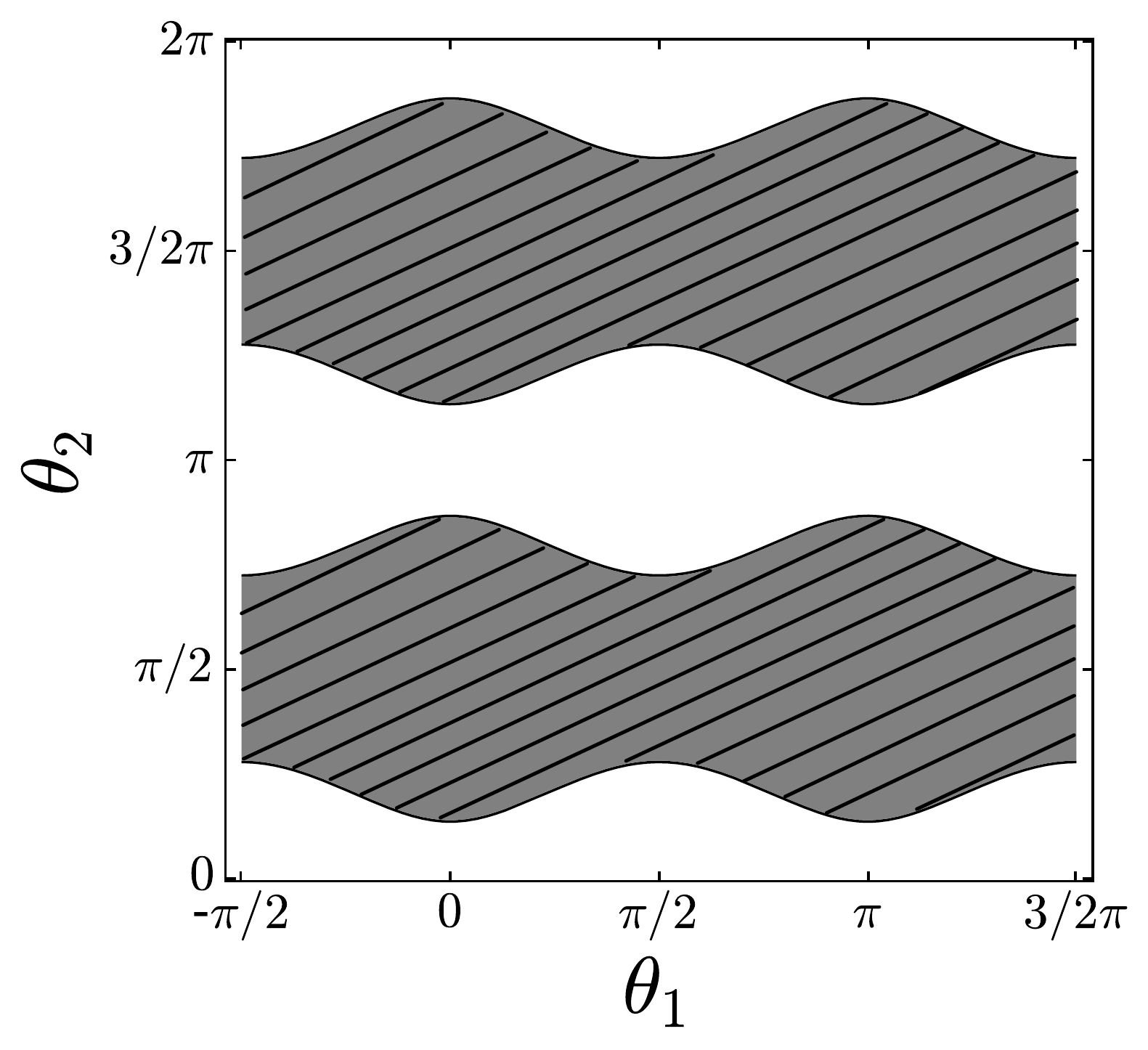}
        \caption{ $ b_1 = 4 $ , $ b_2 = 1 $, $ k_1 =2 $ , and $ k_2 =1.1 $ }
        \label{fig:Z4a}
    \end{subfigure}\qquad\qquad
    ~ 
    \begin{subfigure}[b]{0.45\textwidth}
        \includegraphics[width=\textwidth]{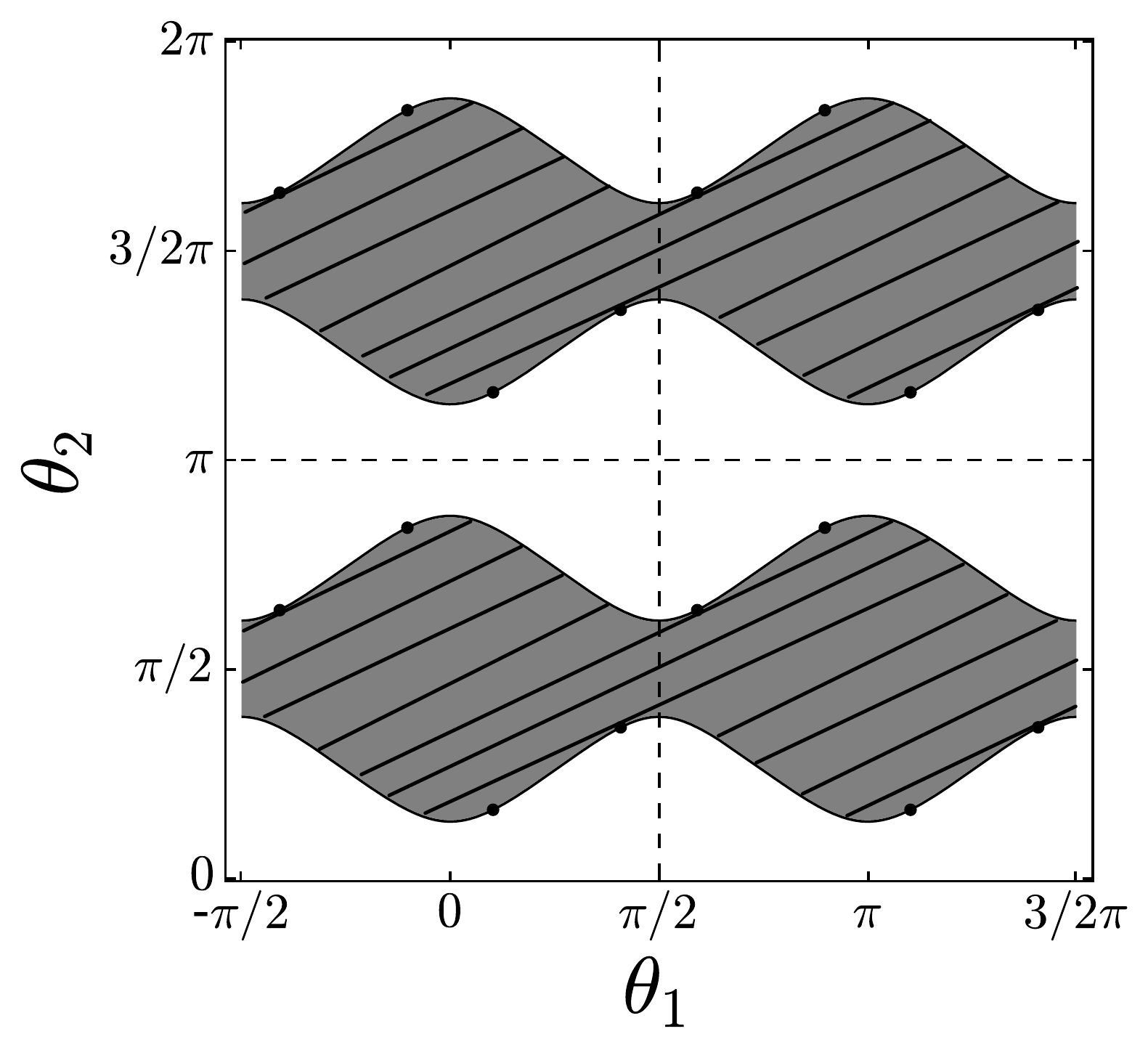}
        \caption{ $ b_1 = 4 $ , $ b_2 = 1 $, $ k_1 =3.4 $ , and $ k_2 =1.2 $   }
        \label{fig:Z4b}
    \end{subfigure}

    \caption{The flow for $ k \in  D_4 $, for two different values of $ k $. 
    }\label{fig:Z4}
\end{figure}

\noindent
Figures \ref{fig:Z2} and \ref{fig:Z4} illustrate the projection of the Suslov flow when $k \in D_2$ and $k \in D_4$, respectively. We notice that Figure \ref{fig:Z4a} corresponds to $k \in D_4^4$, where there is no critical point, while Figure \ref{fig:Z4b} corresponds to $k \in D_4^1 \union D_4^2$. One can understand the periodicity of the Suslov orbits from these projections.

\begin{lemma}\label{lemma:noncriticalperiodic}
 Let $\mc O$ be a non-critical orbit of the Suslov flow on the level set $S_k$ and $\mc O_T$ its projection to the torus. If $\mc O_T \inter \partial U_k$ contains at least $2$ points, then $\mc O$ is periodic.
\end{lemma}
\begin{proof} 
Let $p,q \in \mc O_T \inter \partial U_k$ and let $\ell$ be the component of $\mc O_T$ such that $\partial\ell = \{p,q\}$. By Lemma \ref{lemma:projectionimage}, we see that $(\varphi_k\circ p_k)^{-1}(\ell) \cong S^1$. Since $\mc O$ is connected, we see that $\mc O \cong S^1$, i.e. it is periodic.
\end{proof} 
From the proof, we also see that the torus projection of the closure of a Suslov orbit can have at most two intersection points with $\partial U_k$. The following proposition states that in a large open set of the configuration space, Suslov orbits are generically periodic.

\begin{prop}\label{prop:periodicorbitgalore}
 Let $\displaystyle{Q_1 = \union_{k \not \in \bar{D_1}} S_k}$. Then Suslov orbits in $Q_1$ are generically periodic. 
\end{prop}
\begin{proof} 
Suppose that $\displaystyle{\sqrt{\frac{b_2}{b_1}} \in \QQ}$, then the corresponding linear flow on $T^2$ are periodic. Let $\mc O$ be a Suslov orbit on an smooth level surface $S_k$, then it may not be periodic only if its torus projection $\mc O_T$ contains a critical point in its closure. There is a  finite number of those non-periodic orbits on each level surface, which implies that generic Suslov orbits are periodic. Note that in this case, we do not have to restrict to $Q_1$.
 
 Suppose that $\displaystyle{\sqrt{\frac{b_2}{b_1}} \not \in \QQ}$, then the corresponding linear flow on $T^2$ are not periodic and we restrict the consideration to $k \in Q_1$. For such $k$, $\partial U_k \neq \emptyset$, and $T^2 \setminus U_k$ is an open subset. Any orbit of the corresponding linear flow is dense in $T^2$, and intersects $\partial U_k$ infinitely many times. Since there are only finitely many critical points on each level surface $S_k$, there are only finitely many linear orbits on $T^2$ that intersect with the torus projection of the critical points. By Lemma \ref{lemma:noncriticalperiodic} all the Suslov orbits are periodic, except for a finite number  which connects critical points.
\end{proof} 

We remark that when there is no critical point on a level surface in $Q_1$, e.g. $k \in D_4^4$, all  Suslov orbits on such $S_k$ are periodic. Furthermore, when a Suslov orbit is not periodic, it can be either homoclinic or heteroclinic, e.g. Figure \ref{fig:Z2a} depicts $4$ heteroclinic orbits and $8$ homoclinic orbits, while in Figure \ref{fig:Z4b} there are $16$ homoclinic orbits.

\subsection{Topology of the level surfaces \texorpdfstring{$S_k$}{Sk} via the  Poincar\'e-Hopf theorem}
The Poincar\'e-Hopf theorem \cite{Milnor_Topology_1997} provides a deep link between a purely analytic concept, namely the index of a vector field, and a purely topological one, that is, the Euler characteristic. Recall that the Euler characteristic of a compact connected orientable two dimensional manifold is given by 
\[
    \mathcal{X} = 2 - 2g, 
\]
where $ g $ is the genus, that is the number of ``holes", and that such manifold  is determined, up to an homeomorphism, by its genus. 
The Poincar\'e-Hopf theorem allows us to determine the topology of $ M $ by counting the indices of the zeroes of a vector field on $ M $.

\begin{theorem}[Poincar\'e-Hopf]
 Let $M$ be a compact manifold and let $ v $ be a smooth vector field on $M$ with isolated zeroes. If $ M $ has a boundary, then $v$ is required to point outward at all boundary points.  Then, the sum of the indices at the zeroes of such vector fields is equal to the Euler characteristic of $M$, that is, we have 
\[ 
     \mathcal{X} (M) = \sum _i \operatorname{index } _{ x _i } (v).   
\]
\end{theorem} 

We now use the Poincar\'e-Hopf theorem to give an alternative proof of Proposition \ref{prop:dim3components}.
Since on a compact two manifold the index of a sink, a source, or a center is $ +1 $, and the index of a hyperbolic saddle point is $ - 1 $, the classification of the critical points given in Proposition \ref{prop:crit_points} together with the knowledge of the number of connected components of the manifolds gives the proof for $ \Delta \neq 0 $. For instance, if $ k \in D _2 $, then there are 8 saddle points, so that $ \mathcal{X} (S _k) = - 8 $, and $ g = \frac{ 2-\mathcal{X}  } { 2 } =  5 $. If $ \Delta = 0 $, the critical points are all degenerate and the vector field near the critical points  is given by \eqref{eqn:Delta}. In this case it is easy to see that the index of any critical point is 0, and thus  $ \mathcal{X} (S _k) = 0 $. Since there are two connected components   $ g =  1 $ on each of them. It follows that $ S _k $ is isomorphic to two copies of $ T ^2 $.

In \cite{fernandez2014geometry} a similar approach was used to obtain the topology of Suslov's problem. The main difference is that the authors used an extension of the Poincar\'e-Hopf theorem that applies  to compact manifolds with boundary even when the vector field does not point outward at all boundary points. 

\section{Physical motion}\label{sec:physicalmotion}

 \subsection{Poisson sphere}\label{subsect:imageonPoissonsphere}
   The $2$-sphere $\gamma_1^2 + \gamma_2^2 + \gamma_3^2 = 1$ is known as the \emph{Poisson sphere}.
   Let $ \pi $ be the projection of $ S _k $ onto the Poisson sphere. The {\it  domain of possible motion} (DPM)   corresponding to $ k \in \mathbb{R}  ^2 $ is the set $P_k= \pi (S _k) \subset S ^2  $, that is, it is  the image of the projection of $S_k$ to the Poisson sphere \cite{fomenko2012visual}. 
   If $ p \in S ^2 $ is  a point on the Poisson sphere,  a vector $ v \in \mathbb{R}  ^2 $ such that  $ (p, v) \in S _k $ is said to be an  {\it admissible velocity} at the point $ p \in S ^2 $.  
   A classification  of the possible types of DPMs together with a study of  the set of admissible velocities gives a topological and geometrical description of the mechanical system and it is useful in describing the main features of  the physical motion for various values of $ k $.
   We rewrite the equations as
   \begin{enumerate}
    \item $\displaystyle{\gamma_1^2 = \frac{1}{b_1}(k_1 - m_1^2)}$, 
    \item $\displaystyle{\gamma_2^2 = \frac{1}{b_2}(k_2 - m_2^2)}$
   \end{enumerate}
   and $\gamma_3^2 = 1 - \gamma_1^2 - \gamma_2^2$. The cardinality of the preimage of a point $x = (\gamma_1, \gamma_2, \gamma_3) \in S_k$ is given by the number of pairs of $(m_1, m_2)$ that satisfy the equations $(1)$ and $(2)$ above. 
   Then $x \in P_k$ iff
   \[\gamma_1^2 \varleq \frac{k_1}{b_1} \text{ and } \gamma_2^2 \varleq \frac{k_2}{b_2}\]
   or equivalently we have
   \begin{equation}\label{eq:Poissonsphereimage}
    P_k = \{\gamma_1^2 + \gamma_2^2 + \gamma_3^2 = 1\} \inter \left\{(\gamma_1, \gamma_2, \gamma_3) \in \left[-\sqrt{\frac{k_1}{b_1}}, \sqrt{\frac{k_1}{b_1}}\right] \times \left[-\sqrt{\frac{k_2}{b_2}}, \sqrt{\frac{k_2}{b_2}}\right] \times \RR\right\}
   \end{equation}
   In the interior of $P_k$, we have $m_1 \neq 0$ and $m_2 \neq 0$, which implies that the projection $S_k \to P_k$ is $4$-to-$1$ in the interior.
   
   The region $P_k$ may have boundary components, over which one or both of $m_1$ and $m_2$ vanish. If exactly one of $m_1$ and $m_2$ vanishes, the projection is $2$-to-$1$. If $ m _1 = 0 $, then $ \dot \gamma _1 = 0 $, and  if $ m _2 = 0 $, then $ \gamma _2 = 0 $.  In the case $m_1 = m_2 = 0$, the corresponding points in $P_k$ are corners and the projection is $1$-to-$1$ and $ \dot \gamma _1 = \dot \gamma _2 = \dot \gamma _3 = 0 $. 
   The diagrams below illustrates the regions $P_k$ for various values of $k$.
    
    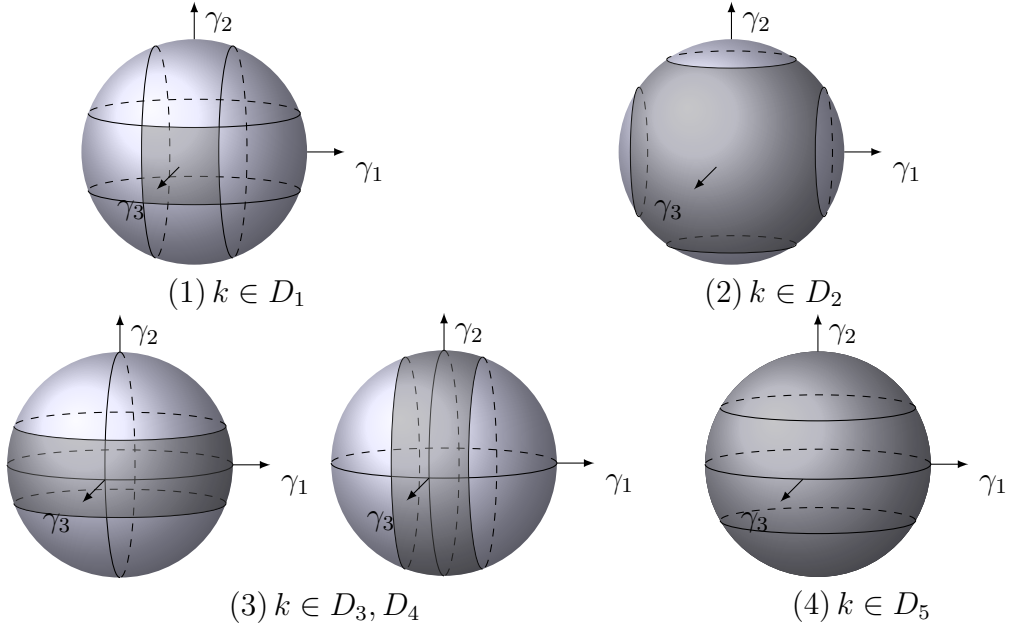
\begin{figure}[!ht]
    \begin{minipage}{0.45\textwidth}
     \begin{center}
      \begin{tikzpicture}[>=latex]
       \shade[ball color=blue!10!white,opacity=1] (0,0) circle (1.5cm);
       \longitude{-20};
       \longitude{20};
       \latitude{20};
       \latitude{-20};
       \begin{scope}
        \clip \regionh{20};
        \clip \regionv{20};
        \fill[ball color=black!10!white,opacity=0.3] \circlez;
       \end{scope}
       \axes
      \end{tikzpicture}
      
      $(1)\, k \in D_1$
     \end{center}
    \end{minipage}
    \begin{minipage}{0.45\textwidth}
     \begin{center}
      \begin{tikzpicture}[>=latex]
       \shade[ball color=blue!10!white,opacity=1] (0,0) circle (1.5cm);
       \longitude{-55};
       \longitude{55};
       \latitude{55};
       \latitude{-55};
       \begin{scope}
        \clip \regionh{55};
        \clip \regionv{55};
        \fill[ball color=black!10!white,opacity=0.3] \circlez;
       \end{scope}
       \axes
      \end{tikzpicture}
     
      $(2)\, k \in D_2$
     \end{center}
    \end{minipage}
    \begin{minipage}{0.6\textwidth}
     \begin{center}
      \begin{tikzpicture}[>=latex]
      \shade[ball color=blue!10!white,opacity=1] (0,0) circle (1.5cm);
      \longitude{0};
      \latitude{20};
      \latitude{0};
      \latitude{-20};
      \begin{scope}
       \clip \regionh{20};
       \fill[ball color=black!10!white,opacity=0.3] \circlez;
      \end{scope}
      \axes
     \end{tikzpicture}
     \begin{tikzpicture}[>=latex]
      \shade[ball color=blue!10!white,opacity=1] (0,0) circle (1.5cm);
      \longitude{0};
      \longitude{20};
      \latitude{0};
      \longitude{-20};
      \begin{scope}  
       \clip \regionv{20};
       \fill[ball color=black!10!white,opacity=0.3] \circlez;
      \end{scope}
      \axes
      \end{tikzpicture}
      
      $(3)\, k \in D_3, D_4$
     \end{center}
    \end{minipage}
    \begin{minipage}{0.3\textwidth}
     \begin{center}
      \begin{tikzpicture}[>=latex]
       \fill[ball color=blue!10,opacity=1] (0,0) circle (1.5 cm);
       \fill[ball color=black!10!white,opacity=0.3]  (0,0) circle (1.5 cm);
       \latitude{30};
       \latitude{0};
       \latitude{-30};
       \axes
      \end{tikzpicture}
     
      $(4)\, k \in D_5$
     \end{center}
    \end{minipage}
\caption{\label{fig:Poissonimage} The domain of possible motion $ P _k $ for various values of $k$. (1) For $ k \in D _1 $, $ P _k $  has the form of two squares on opposite sides of the Poisson sphere. (2) For $ k \in D _2 $, $ P _k $ is a sphere with four holes. (3) For $ k \in D _3, D _4 $, $ P _k $ is a  horizontal or vertical  band  wrapping around the sphere.  (4) For $ k \in D _5 $, $ P _k $ is the whole sphere.}
\end{figure}

Clear pictures emerge when the observations so far are combined. By \eqref{eq:Poissonsphereimage}, the image $P_k$ of $S_k$ on the Poisson sphere is bounded by
 \[\gamma_1 = \pm \sqrt{\frac{k_1}{b_1}} \text{ and } \gamma_2 = \pm \sqrt{\frac{k_2}{b_2}}\]
 which correspond precisely to the following lines on the flat torus $T^2$, as indicated by the dashed lines in Figure \ref{fig:U}:
 \[\theta_1 = \pm \frac{\pi}{2} \text{ and } \theta_2 = 0 \text{ or } \pi\]
 The dashed lines divide $T^2$ into four components, and the projection $ \pi :S_k \to P_k$ restricted to each component is $1$-to-$1$; and the image of shaded region contained in each of the components coincide. 
 The following proposition provides a detailed classification of the DPM for various values of $ k \in \mathbb{R}  ^2 $.
 \begin{prop}\label{prop:projectioncovering}
  Over the interior of the domain of possible motion $P_k$, the projection $\pi :S_k \to P_k$ is $4$-to-$1$. On the boundary components of $P_k$, the projection is $2$-to-$1$, except for over the corners when $k \in D_1$, where it is $1$-to-$1$. Moreover, we have
  \begin{enumerate}
   \item For $k \in D_1$, each torus in $S_k$ is projected onto a component of $P_k$. Each component of $ P _k $ is a square, see figure   \ref{fig:Poissonimage}(1).
   \item For $ k \in D _2 $ the set $ P _k $ is a sphere with four holes as depicted in figure \ref{fig:Poissonimage}(2).
   \item For $k \in D_3$ or $D_4$, the projection $ \pi $  restricted to each torus component of $S_k$  is $2$-to-$1$ in the interior of $P_k$. $ P _k $ is a band wrapping around the Poisson sphere, see figure \ref{fig:Poissonimage}(3).

   \item For $k \in D_5$, the projection $ \pi$ is an isomorphism when restricted to each $S^2$-component of $S_k$, see figure \ref{fig:Poissonimage}(4).
  \end{enumerate}
 \end{prop}

We can now use Proposition \ref{prop:projectioncovering} to understand the physical motion of the rigid body.

If $ k \in D _1 $, the trajectories in each component of $ P _k $, are similar to Lissajous figures (the sum of independent horizontal and vertical oscillations). For each point inside each square there are four admissible velocities, there are two on its sides and one on the vertices.  If $ \sqrt{ \frac{ b _2 } { b _1 } } \not\in \mathbb{Q}  $, then the trajectories is dense in the squares, otherwise they are periodic.  In either case $ \mathbf{E} _3 $ wobbles around the  vertical direction, while  $ \mathbf{E} _1 $  remains close to  horizontal and the wheels remain close to being vertical (see figure \ref{fig:Suslov}). 

If $ k \in D _2 $, almost all the trajectories are periodic except for a finite number of orbits which connect critical points. For points in the interior of $ P _k $ there are four admissible velocities.  There are two admissible velocities on the boundary of $ P _k $.  This means that there are two trajectories  for each point in the interior of $ P _k $ and  each trajectory can be followed in either direction. The physical motion in this case can be distinguished from the previous case since $ \mathbf{E} _3 $ can go from pointing upward to pointing downward.  

If $ k \in D _3 $, the trajectories are confined in a band wrapping around the sphere and alternatively touch the upper and the lower boundary of the band. Since this region is the image of two tori there are two admissible velocities at each point, and a point can move along the trajectories in either direction. In this case $ \mathbf{E} _3 $  performs a complete revolution wobbling about the vertical plane spanned by $ \mathbf{e} _1 $ and $ \mathbf{e} _3 $. The wheels remain close to vertical (see Figure \ref{fig:Suslov}). The case $ k \in D _4 $ is similar. When $b_1 > b_2$, certain subregions of $D_4$ allow homoclinic or heteroclinic orbits. From Figure \ref{fig:Z4b}, we see that in this case, the behaviour of periodic orbits changes drastically on either side of a homoclinic or heteroclinic orbit.

If $ k \in D _5 $, the trajectories are homeomorphic to circles. In this case there are four possible velocities for each point on $ P _k $.  It follows that  there are two trajectories  for each point on the Poisson sphere and  each trajectory can be followed in either direction.  

 

\section*{Acknowledgements}
The authors wish to express their appreciation for helpful discussions with  Luis Garcia-Naranjo  and Dmitry Zenkov.
The research was supported in part by Natural Sciences and Engineering Research Council of Canada (NSERC) Discovery Grants (SH, MS). 
\bibliographystyle{amsplain}
\bibliography{References}
 \end{document}